\newif\ifconf
\conffalse 

\ifconf
\documentclass[envcountsame, oribibl]{llncs}
\else
\documentclass[a4paper,11pt]{article}
\usepackage{a4wide}
\fi

\usepackage{latexsym}
\usepackage{graphicx}
\usepackage{amsfonts}
\usepackage{amsmath} 
\usepackage{amssymb}
\ifconf
\else
\usepackage{amsthm} 
\fi
\usepackage{amscd}
\usepackage{url}
\usepackage{mathtools} 
\usepackage{enumerate}
\usepackage[utf8]{inputenc}

\ifconf
\title{A Direct Proof of the Strong Hanani--Tutte Theorem on the Projective
Plane\thanks{The project was partially supported by the Czech-French collaboration project EMBEDS (CZ:
      7AMB15FR003, FR: 33936TF).
      \'E.~C.~V. was partially supported by the French ANR Blanc project ANR-
12-BS02-005 (RDAM).
 V.~K. was partially supported by the project
      GAUK 926416. V.~K. and M.~T. were partially supported by the project
      GA\v{C}R 16-01602Y.
			P.~P. was supported by the ERC Advanced grant no. 320924.
      Z.~P. was partially supported by Israel Science Foundation grant ISF-768/12.
}
}
\titlerunning{Hanani--Tutte and the Projective Plane}
\author{\'{E}ric Colin de Verdi\`{e}re\inst{1}
\and Vojt\v{e}ch
Kalu\v{z}a\inst{2}
\and Pavel Pat\'{a}k\inst{3} \and Zuzana Pat\'{a}kov\'{a}\inst{3} \and Martin Tancer\inst{2}}
\institute{
D\'epartement d'informatique, \'Ecole normale sup\'erieure,
Paris and CNRS, France
\and
Department of Applied Mathematics,
Charles University in Prague, 
Czech Republic
\and
Einstein Institute of Mathematics, The Hebrew University of Jerusalem, Israel
}

\else
\usepackage{authblk}

\title{A Direct Proof of the Strong Hanani--Tutte Theorem on the Projective
Plane\thanks{The project was partially supported by the Czech-French collaboration project EMBEDS (CZ:
      7AMB15FR003, FR: 33936TF).
      \'E.~C.~V. was partially supported by the French ANR Blanc project ANR-
12-BS02-005 (RDAM).
 V.~K. was partially supported by the project
      GAUK 926416. V.~K. and M.~T. were partially supported by the project
      GA\v{C}R 16-01602Y.
			P.~P. was supported by the ERC Advanced grant no. 320924.
      Z.~P. was partially supported by Israel Science Foundation grant ISF-768/12.
}
}
\author[1]{\'{E}ric Colin de Verdi\`{e}re}
\author[2]{Vojt\v{e}ch Kalu\v{z}a}
\author[3]{Pavel Pat\'{a}k}
\author[3]{Zuzana Pat\'{a}kov\'{a}}
\author[2]{Martin Tancer}
\affil[1]{D\'epartement d'informatique, \'Ecole normale sup\'erieure,
Paris and CNRS, France}
\affil[2]{Department of Applied Mathematics,
Charles University in Prague, 
Czech Republic}
\affil[3]{Einstein Institute of Mathematics, The Hebrew University of Jerusalem, Israel}

\date{}
\fi

\ifconf
\else
 
\fi

\usepackage{color}
\definecolor{red}{rgb}{.8,0,0}

\usepackage{hyperref} 
\definecolor{blue3}{rgb}{.1,.0,.4}
\hypersetup{colorlinks=true, linkcolor=red, urlcolor=blue3, citecolor=blue3, pdfpagemode=UseNone, pdfstartview=FitH, bookmarksopen=true} 
\hypersetup{pdftitle=A Direct Proof of the Strong Hanani--Tutte Theorem on the Projective Plane}
\usepackage[open]{bookmark} 

\ifconf
\else
\newtheorem*{theorem*}{Theorem}
\newtheorem{theorem}{Theorem}
\newtheorem*{lemma*}{Lemma}
\newtheorem{lemma}[theorem]{Lemma}
\newtheorem*{proposition*}{Proposition}
\newtheorem{proposition}[theorem]{Proposition}
\newtheorem*{fact*}{Fact}

\newtheorem*{question*}{Question}

\newtheorem*{corollary*}{Corollary}
\newtheorem{corollary}[theorem]{Corollary}
\newcounter{claimcounter}[theorem]
\numberwithin{claimcounter}{theorem}
\newtheorem*{claim*}{Claim}
\newtheorem{claim}[claimcounter]{Claim}

\theoremstyle{remark}
\newtheorem*{remark*}{Remark}

\theoremstyle{definition}
\newtheorem*{definition*}{Definition}
\newtheorem{definition}[theorem]{Definition}

\numberwithin{equation}{section}
\fi

\ifconf
\spnewtheorem{observation}[theorem]{Observation}{\bfseries}{\itshape}
\spnewtheorem*{observation*}{Observation}{\bfseries}{\itshape}

\spnewtheorem*{proofsketch}{Proof sketch}
{\itshape}{}

\else
\newtheorem*{hypothesis*}{Hypothesis}

\newtheorem*{observation*}{Observation}
\newtheorem{observation}[theorem]{Observation}
\fi

\newcommand{\abs}[1]{\left|#1\right|}

\newcommand{\R}{\mathbb{R}}
\newcommand{\Z}{\mathbb{Z}}

\newcommand{\RP}{\mathbb{R}P}

\DeclareMathOperator{\crno}{cr}

\newcommand{\inn}[1]{#1^+}
\newcommand{\out}[1]{#1^-}
\newcommand{\zin}{(\inn{V}, \inn{E})}
\newcommand{\zout}{(\out{V}, \out{E})}
\newcommand{\gin}{G^{+0}}
\newcommand{\gout}{G^{-0}}
\newcommand{\ar}[2]{\overline{#1#2}}

\begin{document}

\maketitle
\vspace{-1cm}
\begin{abstract}
 We reprove the strong Hanani--Tutte theorem on the projective plane.
 In contrast to the previous proof by Pelsmajer, Schaefer and Stasi, 
 our method is constructive and 
 does not rely on the characterization of forbidden minors,
 which gives hope to extend it to other surfaces.
 Moreover, our approach can be used to provide an efficient algorithm
 turning a Hanani--Tutte drawing on the projective plane into an embedding.
\end{abstract}

\section{Introduction}
A drawing of a graph on a surface is a \emph{Hanani--Tutte drawing} if no two
vertex-disjoint edges cross an odd number of times. We call vertex-disjoint
edges \emph{independent}.

Pelsmajer, Schaefer and Stasi~\cite{PSS09} proved the following theorem via
consideration of the forbidden minors for the projective plane. 
\begin{theorem}[Strong Hanani--Tutte for the projective plane, \cite{PSS09}]\label{thm:main}
 A graph $G$ can be embedded into the projective plane if and only if it admits
 \ifconf\else a Hanani--Tutte drawing \fi on the projective
 plane.\footnote{Of course, the ``only if'' part is trivial.}
\end{theorem}

Our main result is a constructive proof of Theorem~\ref{thm:main}.
The need for a constructive proof is motivated by the strong
Hanani--Tutte conjecture, which states that an analogous result is valid on an
arbitrary (closed) surface. 
This
conjecture is known to be valid only on the sphere (plane) and on the
projective plane. The approach via forbidden minors is relatively simple on the
projective plane; however, this approach does not seem applicable to other surfaces, because there is no reasonable characterization of forbidden minors for them. (Already for the torus or the Klein bottle, the exact list is not
known.)

On the other hand, our approach reveals a number of difficulties that have to be
overcome in order to obtain a constructive proof. If the conjecture is true,
our approach may serve as a basis for its proof on a general surface. If the
conjecture is not true, then our approach may perhaps help to reveal
appropriate structure needed for a construction of a counterexample.

\ifconf
\else
Unfortunately, our approach needs to build an appropriate toolbox for
manipulating with Hanani--Tutte drawings on the projective plane (many tools
are actually applicable to a general surface). This significantly prolongs the
paper. Therefore, we present the main ideas of our approach in the first four
sections of the paper while postponing the technical details to the later sections.
\fi

\paragraph{The Hanani--Tutte theorem on the plane and related results.}
Let us now briefly describe the history of the problem; for complete history and relevant results we refer to
a nice survey by Schaefer~\cite{Schaefer2013}.
Following the work of Hanani~\cite{Hanani34}, Tutte~\cite{Tutte70} made a remarkable observation now known as the (strong) Hanani--Tutte 
theorem: a graph is planar if and only if it admits \ifconf\else
a Hanani--Tutte drawing \fi in the plane.
The theorem has also a parallel history in algebraic topology, where it follows from the ideas of
van Kampen, Flores, Shapiro and Wu~\cite{Kampen33,Wu55,Shapiro57,Levow72}.
 
It is a natural question whether the strong Hanani--Tutte theorem can be extended to graphs on other surfaces;
as we already said before, it has been confirmed only for the projective plane~\cite{PSS09} so far.
On general surfaces, only the weak version~\cite{weak-HT, weak-HT-2} of the theorem is known to be true:
if a graph is drawn on a surface so that every pair of edges crosses an even
number of times\footnote{including 0 times}, 
then the graph can be embedded into the
surface while preserving the cyclic order of the edges at all \ifconf
\else
vertices.\footnote{In fact, the embedding preserves the embedding scheme of
  the graph, where the notion of embedding scheme is a generalization of the
  rotation systems to arbitrary (even non-orientable) surfaces. For more
details on this topic, we refer to~\cite[Chap. 3.2.3]{Gross87}, where
embedding schemes are called rotation systems and our rotation systems are
called pure.} \fi
Note that in the strong version we require that only independent edges cross even number of times,
while in the weak version this condition has to hold for all pairs of edges.

We remark that other variants of the Hanani--Tutte theorem generalizing the notion of embedding in the plane have also been considered. For instance, the strong Hanani--Tutte theorem was proved for partially embedded graphs~\cite{SchaeferTToP} and both weak and strong Hanani--Tutte theorem were proved also for $2$-clustered graphs~\cite{Fulek2015clustered}.

The strong Hanani--Tutte theorem is important from the algorithmic point of
view, since it implies the Tr\'{e}maux crossing theorem,
which is used to prove de Fraysseix-Rosenstiehl's planarity criterion~\cite{Fraysseix85}.
This criterion has been used to justify the linear time planarity algorithms including the Hopcroft-Tarjan~\cite{Hopcroft74} and the Left-Right~\cite{Fraysseix12} algorithms.
For more details we again refer to \cite{Schaefer2013}.
\ifconf \else \medskip \fi

One of the reasons why the strong Hanani--Tutte theorem is so important is that it turns planarity question into 
a system of linear equations. For general surfaces, the question whether there
exists a Hanani--Tutte drawing of $G$ leads to a system of quadratic
equations~\cite{Levow72} over $\Z_2$. If the strong Hanani--Tutte theorem is
true for the surface, any solution to the system then serves as a certificate
that $G$ is embeddable.  Moreover, if the proof of the Hanani--Tutte theorem is
constructive, it gives a recipe how to turn the solution into an actual
embedding.  Unfortunately, solving systems of quadratic equations is
NP-complete.

For completeness we mention that for each surface there
exists a polynomial time algorithm that decides whether a graph 
can be embedded into that surface~\cite{Mohar99, Kawarabayashi08};
however, the hidden constant depends exponentially on the genus.

\ifconf\else \medskip \fi
The original proofs of the strong Hanani--Tutte theorem in the plane used
Kuratowski's theorem~\cite{Kuratowski30}, and therefore are non-constructive. In 2007,
Pelsmajer, Schaefer and \v{S}tefankovi\v{c}~\cite{PSS07} published a
constructive proof. They showed a sequence of moves that change \ifconf
\else a Hanani--Tutte drawing \fi into an embedding.

A key step in their proof is their Theorem~2.1.  
We say that an edge is \emph{even} if it crosses every other edge
an even number of times (including the adjacent edges). 

\begin{theorem}[Theorem~2.1 of~\cite{PSS07}]
\label{t:2.1}
If $D$ is a drawing of a graph $G$ in the plane, and $E_0$ is the set of even
edges in $D$, then
$G$ can be drawn in the plane so that no edge in $E_0$ is involved in an
intersection and there are
no new pairs of edges that intersect an odd number of times.
\end{theorem}

Unfortunately, an analogous result is simply not true on other surfaces, as is shown in~\cite{weak-HT-2}. In
particular, this is an obstacle for a constructive proof of
Theorem~\ref{thm:main}. \ifconf
\else

\paragraph[\texorpdfstring{Our approach---replacement of Theorem~2.1
in~\cite{PSS07}.}{Replacement of Theorem~2.1 in PSS.}]{Our approach---replacement of Theorem~2.1 in~\cite{PSS07}.}

The key step of our approach is to provide a suitable replacement of
Theorem~2.1 in~\cite{PSS07} (Theorem~\ref{t:2.1}); see also Lemma~3 in~\cite{Fulek2012AdjCrossings}. For a description of this
replacement, let us focus on the following simplified setting.

Let us consider the case that we
have a graph $G$ with a Hanani--Tutte drawing $D$ on the sphere $S^2$. 
Let $Z$ be a cycle of $G$ which is \emph{simple}, that is, drawn without
self-intersections, and such that every edge of $Z$ is even. Theorem~\ref{t:2.1} then implies that $G$ can be
redrawn so that $Z$ is free of crossings without introducing new pairs of edges
crossing oddly.

\begin{figure}
\begin{center}
\includegraphics{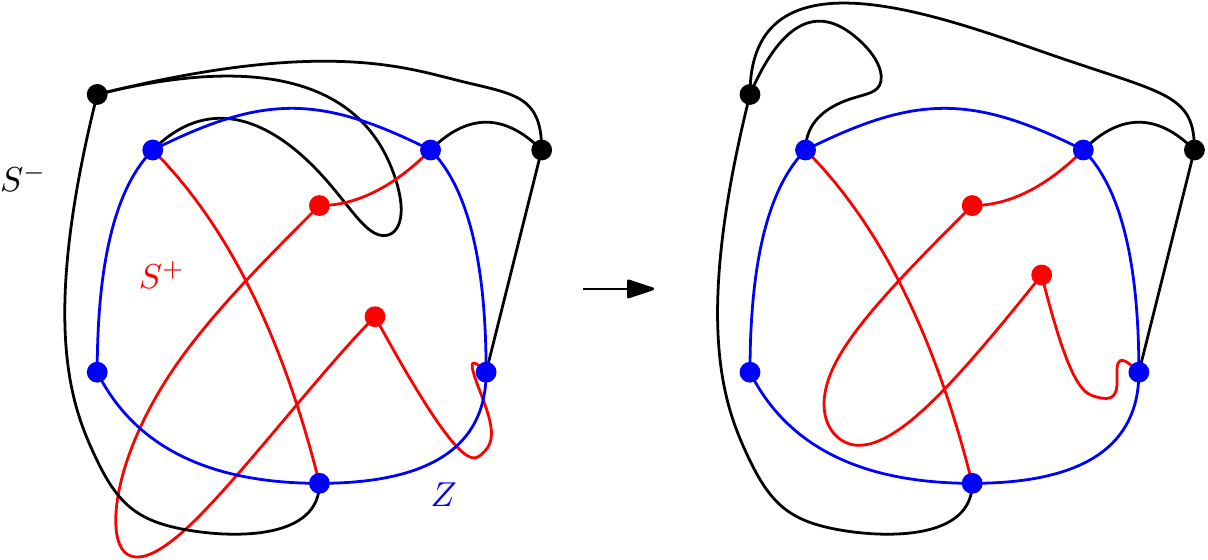}
\caption{Separating the outside (in black) and the inside (in red).}
\label{f:in_out}
\end{center}
\end{figure}

Actually, a detailed inspection of the proof in~\cite{PSS07} reveals something
slightly stronger in this setting. The drawing of $Z$ splits the plane into two
parts that we call the \emph{inside} and the \emph{outside}. This in turn
splits $G$ into two parts. The inside part consists of vertices that are inside
$Z$ and of the edges that
have either at least one endpoint inside $Z$, or they have both endpoints on $Z$
and they enter the inside of $Z$ next to both endpoints. The outside part is
defined analogously. Because we have started with a Hanani--Tutte drawing, it
is easy to check that every vertex and every edge is on $Z$ or inside or
outside. The proof of Theorem~\ref{t:2.1} in~\cite{PSS07} then implies that the
inside and the outside may be fully separated in the drawing; see
Fig.~\ref{f:in_out}. Actually, this can be done even by a continuous motion---if
the drawing is considered on the sphere (instead of the plane).

The trouble on $\RP^2$ is that it may not be possible to separate the
outside and the inside by a continuous motion (of each of the parts
separately). This is demonstrated by a projective-planar drawing of $K_5$ in
Fig.~\ref{f:K_5_example}, left. (The symbol `$\otimes$' stands for the
crosscap in the picture.)

\begin{figure}
\begin{center}
\includegraphics{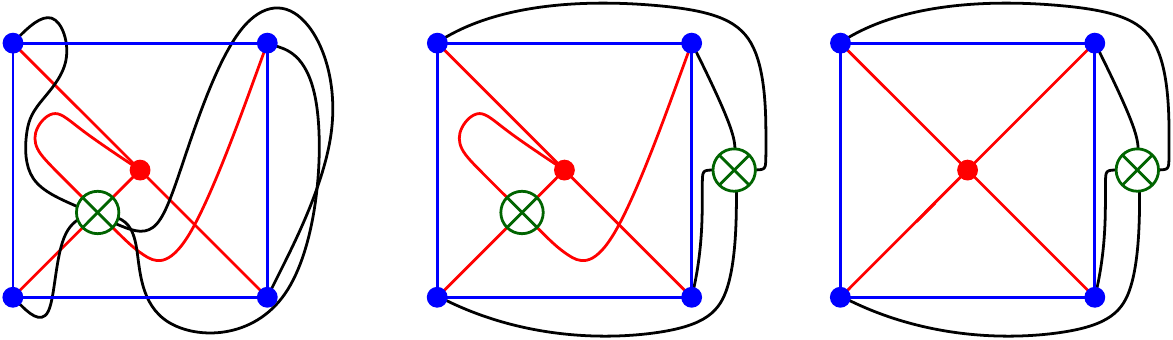}
\caption{Projective-planar drawing of $K_5$ where the outside and the inside cannot be 
separated by a continuous motion (right) and a solution by
duplicating the crosscap (middle) and removing one of them (right).}
\label{f:K_5_example}
\end{center}
\end{figure}

It would actually help significantly if we were allowed to duplicate the
crosscap as in Fig.~\ref{f:K_5_example}, middle. However, the problem is that
we cannot afford raising the genus. On the other hand, if we give up on a continuous
motion, we may observe that the inside vertices and edges in
Fig.~\ref{f:K_5_example}, middle, may be actually redrawn in a planar way if we
remove the `inside' crosscap. This step changes the homotopy/homology type of
many cycles in the drawing.

Our main technical contribution is to show that it is not a coincidence that
this simplification of the drawing in Fig.~\ref{f:K_5_example} was possible. We will show that it is always
possible to redraw one of the sides without using the `duplicated' crosscap.
The precise statement is given by Theorem~\ref{t:black_box}.

\paragraph{The remainder of the proof.} 
As we mentioned above, Theorem~\ref{t:2.1} is a key ingredient in the proof of
the strong Hanani--Tutte theorem in the plane. The rough idea is to find a
suitable order on some of the cycles of the graph so that Theorem~\ref{t:2.1}
can be used repeatedly on these cycles eventually obtaining a planar drawing. A
detailed proof of Pelsmajer, Schaefer and \v{S}tefankovi\v{c} uses an induction
based on this idea. 

Similarly, we use Theorem~\ref{t:black_box} in an inductive proof of
Theorem~\ref{thm:main}. The details in our setting are more
complicated, because we have to take care of two types of cycles in the
graph based on their homological triviality. We also need to put more
effort to set up the induction in a suitable way for using
Theorem~\ref{t:black_box}, because our setting for Theorem~\ref{t:black_box} is
slightly more restrictive than the setting of Theorem~\ref{t:2.1}.
\fi

\ifconf
\else

\paragraph{Organization of the paper.}
In Sect.~\ref{s:ht_drawings} we describe Hanani--Tutte drawings on the
projective plane and their properties. There we also set up several tools for
modifications of the drawings. In particular, we describe how to transform
the Hanani--Tutte drawings on $\RP^2$ into drawings on the sphere satisfying a
certain additional condition. This helps significantly in several cases with
manipulating these drawings.
In Sect.~\ref{s:blackbox} we describe the precise statement of
Theorem~\ref{t:black_box}. We also provide a proof of this theorem in that
section, however, we postpone the proofs of many auxiliary results to later
sections. 
In Sect.~\ref{s:induction} we prove Theorem~\ref{thm:main} using
Theorem~\ref{t:black_box} and some of the auxiliary results from
Sect.~\ref{s:blackbox}.
The remaining sections are devoted to the missing proofs of auxiliary results.
\fi

\section{Hanani--Tutte Drawings}
\label{s:ht_drawings}

In this section, we consider Hanani--Tutte drawings of graphs on the
sphere and on the projective plane. We use the standard notation from graph
theory. Namely, if $G$ is a graph, then $V(G)$ and $E(G)$ denote the set of
vertices and the set of edges of $G$, respectively. Given a vertex $v$ or an
edge $e$, by $G - v$ or $G - e$ we denote the graph obtained from $G$ by
removing $v$ or $e$, respectively.

\ifconf\else
\fi

Regarding drawings of graphs, first, let us recall a few standard definitions considered on an arbitrary surface.
We put the standard general position assumptions on the drawings. That is, we
consider only drawings of graphs on a surface such that no edge contains a
vertex in its interior and every pair of edges meets only in a finite number of
points, where they \emph{cross} transversally. However, we allow three or more
edges meeting in a single \ifconf
\else 
point (we do not mind them because we study the
pairwise interactions of the edges only). Let us also mention that, in all this
paper, we can assume that in every drawing, every edge is free of
self-crossings. Indeed, we can remove any self-crossing without changing the
image of the edge, except in a small neighborhood of the self-crossing.
\fi

\ifconf
\else
We recall from the introduction that two edges are independent if they do not
share a vertex.
Given a surface $S$ and a graph $G$, 
a \emph{(strong) Hanani--Tutte drawing} of
$G$ on $S$ is a drawing of $G$ on $S$ such that every pair of independent edges
crosses an even number of times. 
We will often abbreviate the term (strong) Hanani--Tutte drawing to \emph{HT-drawing}.

\fi

\ifconf\else
\paragraph{Crossing numbers.} \fi
Let $D$ be a drawing of a graph $G$ on a surface $S$. Given two distinct
edges $e$ and $f$ of $G$ by $\crno(e,f) = \crno_{D}(e,f)$ we denote the number of
crossings between $e$ and $f$ in $D$ modulo 2.
We say that an edge $e$ of $G$
is \emph{even} if $\crno(e,f) = 0$ for any $f \in E(G)$ distinct from $e$.
We emphasize that we consider the crossing number as an element of $\Z_2$ and
all computations throughout the paper involving it are done in $\Z_2$.

\paragraph[\texorpdfstring{HT-drawings on $\RP^2$.}{HT-drawings on the projective plane.}]{HT-drawings on $\RP^2$.}
It is convenient for us to set up some conventions for working with the HT-drawings on
the (real) projective plane, $\RP^2$. There are various ways to represent
$\RP^2$. Our convention will be the following: we consider the sphere $S^2$ and a disk
(2-ball) $B$ in it. We remove the interior of $B$ and identify the opposite points on the boundary $\partial B$. 
This way, we obtain a representation of $\RP^2$.
Let $\gamma$ be the curve coming from $\partial B$ after the identification. We
call this curve a \emph{crosscap}. It is a homologically (homotopically)
non-trivial simple cycle (loop) in $\RP^2$, and conversely, any homologically
(homotopically) nontrivial simple cycle (loop) may serve as a crosscap up to a self-homeomorphism of
$\RP^2$. In drawings, we use the symbol $\otimes$ for the crosscap coming from
the removal of the disk `inside' this symbol. \ifconf\else We also use this
symbol for ends of proofs.\fi

Given an HT-drawing of a graph on $\RP^2$, it can be slightly shifted so that it
meets the crosscap in a finite number of points and only transversally, still keeping
the property that we have an HT-drawing. Therefore, we may add to our conventions 
that this is the case for our HT-drawings on $\RP^2$. 

Now, we consider a map $\lambda\colon E(G) \to \Z_2$. For an edge $e$, we let
$\lambda(e)$ be the number of crossings of $e$ and the crosscap $\gamma$
modulo 2. 
We emphasize that $\lambda$ depends on the choice of the crosscap. \ifconf\else Afterwards, it 
will be useful to alter $\lambda$ via so-called vertex-crosscap 
switches, which we will explain a bit later. \fi

Given a (graph-theoretic) cycle $Z$ in $G$, we can distinguish whether $Z$ is
drawn as a homologically nontrivial cycle by checking the value $\lambda(Z) := \sum\lambda(e)
\in \Z_2$ where the sum is over all edges of $Z$. The cycle $Z$ is
homologically nontrivial if and only if $\lambda(Z) = 1$. In particular,
it follows that $\lambda(Z)$ does not depend on the choice of the crosscap.

\paragraph[\texorpdfstring{Projective HT-drawings on $S^2$.}{Projective HT-drawings on the sphere.}]{Projective HT-drawings on $S^2$.}

Let $D$ be an HT-drawing of a graph $G$ on $\RP^2$. 
It is not hard to deduce a
drawing $D'$ of the same graph on $S^2$ such that every pair $(e,f)$ of \emph{independent}
edges satisfies $\crno(e,f) = \lambda(e)\lambda(f)$. Indeed, it is sufficient
to `undo' the crosscap, glue back the disk $B$ and then let the edges
intersect on $B$. 
\ifconf
\else
See the two leftmost pictures in Fig.~\ref{f:redraw_crosscap}.
This motivates the following definition.
\fi

\begin{definition}
  \label{d:projective_HT}
  Let $D$ be a drawing of a graph $G$ on $S^2$ and $\lambda\colon E(G) \to
  \Z_2$ be a function. Then the pair $(D, \lambda)$ is a \emph{projective
  HT-drawing of $G$ on $S^2$} if $\crno(e,f) = \lambda(e)\lambda(f)$ for any pair
  of independent edges $e$ and $f$ of $G$. \ifconf \else (If $\lambda$ is sufficiently 
  clear from the context, we say that $D$ is a projective HT-drawing
  of $G$ on $S^2$.) \fi
\end{definition}

\ifconf\else
\begin{figure}
\begin{center}
\includegraphics{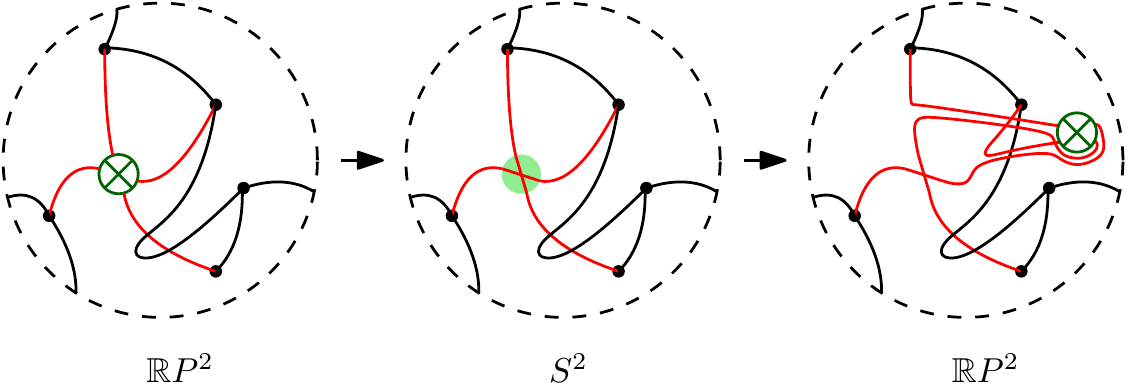}
\caption{Transformations between HT-drawings on $\RP^2$ and projective
HT-drawings on $S^2$.}
\label{f:redraw_crosscap}
\end{center}
\end{figure}
\fi

It turns out that a projective HT-drawing on $S^2$ can also be transformed to
an HT-drawing on $\RP^2$.

\ifconf
\else
\begin{lemma}
\label{l:S2_to_RP2}
Let $(D,\lambda)$ be a projective HT-drawing of a graph $G$ on $S^2$.
Then there is an HT-drawing $D_{\otimes}$ of
  $G$ on $\RP^2$ such that $\crno_{D_{\otimes}}(e,f) = \crno_{D}(e,f) +
  \lambda(e)\lambda(f)$ for any pair of distinct edges of $G$, possibly
  adjacent. In addition, if $e$ and $f$ are arbitrary two edges such that
  $\lambda(e) = \lambda(f) = 0$ and $D(e)$ and $D(f)$ are disjoint; then
  $D_\otimes(e)$ and $D_\otimes(f)$ are disjoint as well.
\end{lemma}

\begin{proof}
It is sufficient to consider a small disk $B$ which does not
intersect $D(G)$, replace it with a crosscap and redraw the edges
$e$ with $\lambda(e) = 1$ appropriately as described below. (Follow the two
pictures on the right in Fig.~\ref{f:redraw_crosscap}.) 
From each edge $e$ with $\lambda(e) = 1$, we pull a
thin `finger-move' towards the crosscap which intersects every other edge in
pairs of intersection points. Then we redraw the edge in a close neighbourhood
of the crosscap as indicated in Fig.~\ref{f:finger_crosscap}. After this
redrawing, each edge $e$ such that $\lambda(e) = 1$ passes over the crosscap
once and each edge $e$ with $\lambda(e) = 0$ does not pass over it. This agrees
with our original definition of $\lambda$ for HT-drawings on $\RP^2$. In
addition, we indeed obtain an HT-drawing on $\RP^2$ with
$\crno_{D_{\otimes}}(e,f) = \crno_{D}(e,f) +
  \lambda(e)\lambda(f)$, because in the last step 
we introduce one more crossing among pairs of edges $e$, $f$ such that
$\lambda(e) = \lambda(f) = 1$.
\ifconf
\qed
\fi
\end{proof}
\begin{figure}
\begin{center}
\includegraphics{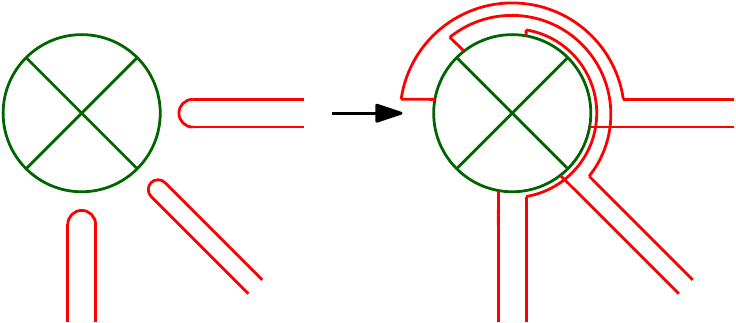}
\caption{Redrawing the finger-moves around the crosscap.}
\label{f:finger_crosscap}
\end{center}
\end{figure}

In summary, Lemma~\ref{l:S2_to_RP2} together with the previous discussion provide
us with two viewpoints on the Hanani--Tutte drawings.

\begin{corollary}
\label{c:projective_drawings}
A graph $G$ admits a projective HT-drawing on $S^2$ 
(with respect to some
function $\lambda\colon E(G) \to \Z_2$)
if and only if it admits an HT-drawing
on $\RP^2$.
\end{corollary}
\fi

The main strength of \ifconf\else
Corollary~\ref{c:projective_drawings}  \fi relies in the fact
that in projective HT-drawings on $S^2$ we can ignore the actual geometric position of
the crosscap and work in $S^2$ instead, which is simpler. This is especially
helpful when we need to merge two drawings. \ifconf\else
On the other hand, it
turns out that for our approach it will be easier to perform certain parity counts
in the language of HT-drawings on $\RP^2$.
\fi

In order to distinguish the usual HT-drawings on $S^2$ from the projective HT-drawings, we will sometimes refer to the former as to the \emph{ordinary} HT-drawings on $S^2$.

\paragraph{Nontrivial walks.}

Let $(D, \lambda)$ be a projective HT-drawing of a graph $G$
and $\omega$ be a walk in $G$. We define 
$\lambda(\omega) := \sum_{e\in E(\omega)} \lambda(e)$ where $E(\omega)$ is the multiset of edges appearing in
$\omega$. Equivalently, it is sufficient to consider only the edges appearing
an odd number of times in $\omega$, because $2\lambda(e) = 0$ for any edge $e$.
We say that $\omega$ is \emph{trivial} if $\lambda(\omega) = 0$ and
\emph{nontrivial} otherwise. \ifconf \else

\fi
We often use this terminology in special cases when $\omega$ is an edge, a
path, or a cycle. \ifconf
\else
In particular, a cycle $Z$ is trivial if and only if it is drawn as a homologically 
trivial cycle in the corresponding drawing $D_\otimes$ of $G$ on $\RP^2$ from
Lemma~\ref{l:S2_to_RP2}.
\fi

\ifconf
\else
Given two homologically nontrivial cycles on $\RP^2$ it is well known that they
must cross an odd number of times (assuming they cross at every intersection).
This fact is substantiated by Lemma~\ref{lem:2cycles} later on.
However, we 
first present a weaker version of this statement in the setting of projective 
HT-drawings, which we need sooner.

\begin{lemma}
\label{l:disjoint_cycles_projective}
Let $(D,\lambda)$ be a projective HT-drawing of a graph $G$ on $S^2$. Then $G$
does not contain two vertex-disjoint nontrivial cycles.
\end{lemma}

\begin{proof}
For contradiction, let $Z_1$ and $Z_2$ be two vertex-disjoint nontrivial cycles in
$G$. That is, $Z_1$ as well as $Z_2$ contains an odd number of nontrivial edges.
Therefore, there is an odd number of pairs $(e_1, e_2)$ of nontrivial edges
where $e_1 \in Z_1$ and $e_2 \in Z_2$. According to
Definition~\ref{d:projective_HT}, $Z_1$ and $Z_2$ must have an odd number of
crossings. But this is impossible for two cycles in the plane which cross at
every intersection (in $D$).
\ifconf
\qed
\fi
\end{proof}
\fi

\ifconf
\else
\paragraph{Vertex-edge and vertex-crosscap switches.}

Let $D$ be a drawing of a graph $G$ on $S^2$. Let us consider a vertex $v$ and an edge
$e$ of $G$ such that $v$ is not incident to $e$. We modify the drawing $D$
into drawing $D'$ so that we pull a thin finger from the interior of $e$
towards $v$ and we let this finger pass over $v$. We say that $D'$ is obtained
from $D$ by the \emph{vertex-edge switch} $(v, e)$.\footnote{Another name for
the \emph{vertex-edge switch} is the \emph{finger-move} common mainly in
topological context in higher dimensions.} If we have an edge $f$
incident to $v$, then the crossing number $\crno(e,f)$ of this pair changes (from
$0$ to $1$ or vice versa), but it does not change for any other pair, because the
`finger' intersects the other edges in pairs.

Now, let $(D,\lambda)$ be a projective HT-drawing of $G$ on $S^2$.
It is very useful to alter $\lambda$ at the
cost of redrawing $G$. Given a vertex $v$, we perform the vertex-edge switches
$(v, e)$ for all edges $e$ not incident to $v$ such that $\lambda(e) = 1$
obtaining a drawing $D'$. We also introduce a new function
$\lambda'\colon E(G) \to \Z_2$ derived from $\lambda$ by switching the value of
$\lambda$ on all edges of $G$ incident to $v$. In this case, we say that $D'$
(and $\lambda'$) is obtained by the \emph{vertex-crosscap switch} over $v$.\footnote{In the case of drawings on $\RP^2$, a vertex-crosscap switch corresponds to passing
the crosscap over $v$, which motivates our name. On the other hand, it is beyond
our needs to describe this correspondence exactly.} 
It yields again an HT-drawing.

\begin{lemma}\label{l:vertex_crosscap_switch}
  Let $(D, \lambda)$ be a projective HT-drawing of $G$ on $S^2$.
  Let $D'$ and $\lambda'$ be obtained
  from $D$ and $\lambda$ by a vertex-crosscap switch. Then $(D', \lambda')$ is a projective
HT-drawing of $G$ on $S^2$.
\end{lemma}

\begin{proof}
It is routine to check that $\crno_{D'}(e,f) = \lambda'(e)\lambda'(f)$ for any
pair of independent edges $e$ and $f$.

Indeed, let $v$ be the vertex inducing the switch. If neither $e$ nor $f$ is
incident to $v$, then 
\[
\crno_{D'}(e,f) = \crno_{D}(e,f) = \lambda(e)\lambda(f) =
\lambda'(e)\lambda'(f).
\]
It remains to consider the case that one of the edges, say $e$, is incident to
$v$. Note that $\lambda(e) = 1- \lambda'(e)$ and $\lambda(f) = \lambda'(f)$ in
this case.

If $\lambda(f) = 0$, then
\[
\crno_{D'}(e,f) = \crno_{D}(e,f) = \lambda(e)\lambda(f) = 0 = 
\lambda'(e)\lambda'(f).
\]

Finally, if $\lambda(f) = 1$, then
\[
\crno_{D'}(e,f) = 1 - \crno_{D}(e,f) = 1 - \lambda(e)\lambda(f) = \lambda(f) -
\lambda(e)\lambda(f) = 
\lambda'(e)\lambda'(f).
\]
\ifconf
\qed
\fi
\end{proof}

We also remark that a vertex-crosscap switch keeps the triviality or nontriviality of cycles.
Indeed, let $Z$ be a cycle. If $Z$ avoids $v$, then $\lambda(Z) = \lambda'(Z)$
since $\lambda(e) = \lambda(e')$ for any edge $e$ of $Z$. If $Z$ contains $v$,
then $\lambda(Z) = \lambda'(Z)$ as well since $\lambda(e) \neq \lambda'(e)$ for
exactly two edges of $Z$.
\fi

\ifconf
\else
\paragraph{Planarization.}

As usual, let $(D, \lambda)$ be a projective HT-drawing of $G$ on $S^2$.
Now let us consider a subgraph $P$ of $G$ such
that every cycle in $P$ is trivial. Then $P$ essentially behaves as a planar
subgraph of $G$, which we make more precise by the following lemma.

\begin{lemma}
  \label{l:planarize}
  Let $(D, \lambda)$ be a projective HT-drawing of $G$ on $S^2$
  and let $P$ be a subgraph of $G$ such 
that every cycle in $P$ is trivial. Then there is a set $U \subseteq V(P)$
with the following property. Let $(D_U,\lambda_U)$ be obtained from $(D,\lambda)$ by the vertex-crosscap switches over all vertices of $U$ (in any
order). Then $(D_U,\lambda_U)$ is a projective HT-drawing of $G$ on $S^2$ and $\lambda_U(e) = 0$ for any edge $e$ of $E(P)$. 
\end{lemma}

\begin{proof}
The drawing $(D_U,\lambda_U)$ is a projective HT-drawing by Lemma~\ref{l:vertex_crosscap_switch}.
Let $F$ be a spanning forest of~$P$, the union of spanning trees of each
connected component of~$P$, rooted arbitrarily.  We first make $\lambda(e)=0$
for each edge of~$F$, as follows: do a breadth-first search on each tree
in~$F$; when an edge~$e\in F$ with $\lambda(e)=1$ is encountered, perform a
vertex-crosscap switch on the vertex of~$e$ farther from the root of the tree.
Let $\lambda_U$ be the resulting map, which is zero on the edges of~$F$.  Each
edge~$e$ in~$E(P)\setminus E(F)$ belongs to a cycle~$Z$ such that $Z - e\subseteq F$.
Since $\lambda_U(Z)=\lambda(Z)=0$, we have $\lambda_U(e)=0$ as well.
\ifconf
\qed
\fi
\end{proof}
\fi

\section{Separation Theorem}
\label{s:blackbox}

In this section, we state the \ifconf\else separation theorem \fi announced in the introduction.
\ifconf
\else

As it was explained in the introduction, a simple cycle $Z$ such that every
edge of $Z$ is even (in a drawing) splits the graph into the outside and the
inside. We first introduce a notation for this splitting.
\fi

\begin{definition}
\label{d:in_out}
Let $G$ be a graph and $D$ be a drawing of $G$ on $S^2$. Let us assume that $Z$
is a cycle of $G$ such that every edge of $Z$ is even and it is drawn as a
simple cycle in $D$. Let $S^+$ and $S^-$ be the two components of $S^2
\setminus D(Z)$. We call a vertex $v \in V(G) \setminus V(Z)$ an \emph{inside}
vertex if it
belongs to $S^+$ and an \emph{outside} vertex otherwise. Given an edge $e = uv
\in E(G) \setminus E(Z)$, we say that $e$ is an \emph{inside} edge if either
$u$ is an
inside vertex or if $u \in V(Z)$ and $D(e)$ points locally to $S^+$ next to
$D(u)$. Analogously we define an \emph{outside} edge.\footnote{It turns out that every
edge $e \in E(G) \setminus E(Z)$ is either an outside edge or an inside edge, because
every edge of $Z$ is even.}
We let $V^+$ and $E^+$ be the sets of the inside vertices and the inside
edges, respectively. Analogously, we define $V^-$ and
$E^-$.
We also define the graphs $G^{+0} := (V^+ \cup V(Z), E^+
\cup
E(Z))$ and $G^{-0} := (V^- \cup V(Z), E^- \cup E(Z))$.
\end{definition}

Now, we may formulate our main technical tool---the separation theorem for
projective HT-drawings.

\begin{theorem}

\label{t:black_box}
Let $(D,\lambda)$ be a projective HT-drawing of a $2$-connected graph $G$ on $S^2$ and $Z$ a cycle
of $G$ that is simple in $D$ and such that every edge of $Z$ is even. Moreover, we assume that every edge $e$ of $Z$ is
trivial, that is, $\lambda(e) = 0$. Then
there is a projective HT-drawing $(D', \lambda')$ of $G$ on $S^2$ 
satisfying the following
properties.

\begin{itemize}
        \item The drawings $D$ and~$D'$ coincide on $Z$;
	\item the cycle $Z$ is \ifconf \else completely \fi free of crossings and all of its
          edges
	  are trivial in $D'$;
	\item $D'(\gin)$ is contained in $\inn{S} \cup D'(Z)$;
	\item $D'(\gout)$ is contained in $\out{S} \cup D'(Z)$; and
	\item either all edges of $\gin$ or all edges of $\gout$ are trivial
	  (according to $\lambda'$);
	  that is, at least one of the drawings $D'(\gin)$ or $D'(\gout)$ is an ordinary HT-drawing on $S^2$.
\end{itemize}
\end{theorem}

\ifconf
\else
The assumption that $G$ is $2$-connected is not essential for the proof of
Theorem~\ref{t:black_box}, but it will slightly simplify some of the steps. (For
our application, it will be sufficient to prove the $2$-connected case.)
\fi

In the remainder of this section, we describe the main ingredients of the
proof of Theorem~\ref{t:black_box} and we also derive this theorem from the
ingredients. We will often encounter the setting when $G$,  $(D,\lambda)$ and $Z$
satisfy the assumptions of Theorem~\ref{t:black_box}. Therefore, we say that $G$,
$(D,\lambda)$ and $Z$ satisfy the \emph{separation assumptions} if (1) $G$ is a
$2$-connected graph; (2) $(D,\lambda)$ is a projective HT-drawing of $G$;
(3) $Z$ is a cycle in $G$ drawn as a simple cycle in $D$; (4) every edge of $Z$ is
even in $D$ and trivial.
\ifconf\else
\fi

\paragraph{Arrow graph.}
From now on, let us fix $G$, $(D,\lambda)$ and $Z$ satisfying the separation
assumptions. This also fixes the distinction between the outside and the
inside.

\begin{definition}\label{d:bridge}
  A \emph{bridge}~$B$ of~$G$ (with respect to~$Z$) is a subgraph of~$G$
  that is either an edge not in~$Z$ but with both endpoints in~$Z$ (and its
  endpoints also belong to~$B$), or a connected component of $G-V(Z)$
  together with all edges (and their endpoints in~$Z$) with one endpoint in
  that component and the other endpoint in~$Z$.\ifconf
\else \ (This is a standard
  definition; see, e.g., Mohar and Thomassen~\cite[p.~7]{mt-gs-01}.) \fi

We say that $B$ is an \emph{inside bridge} if it is a subgraph of
$G^{+0}$, and an \emph{outside bridge} if it is a subgraph of $G^{-0}$
(every bridge is thus either an inside bridge or an outside bridge).

A walk~$\omega$ in~$G$ is a \emph{proper walk} if no vertex in $\omega$ belongs to
$V(Z)$, except possibly its endpoints, and no edge of $\omega$ belongs to
$E(Z)$. In particular, each proper walk belongs to a single bridge.

\end{definition}

Since we assume that $G$ is $2$-connected, every inside bridge contains at
least two vertices of $Z$. 
The bridges induce partitions of $E(G)\setminus E(Z)$ and of
$V(G)\setminus V(Z)$. \ifconf\else
See Fig.~\ref{f:inside_bridges}.

\begin{figure}
\begin{center}
\includegraphics{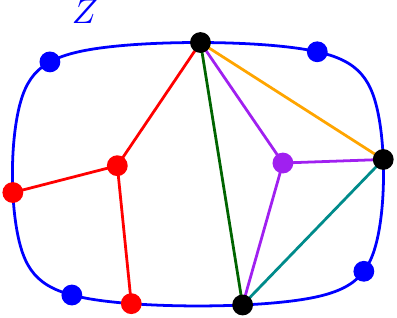}
\caption{An example of a graph with five inside bridges---marked by
different colours. The vertices that belong to several inside bridges are in
black.}
\label{f:inside_bridges}
\end{center}
\end{figure}

\fi

We want to record which pairs of vertices on $V(Z)$ are connected with a
nontrivial and proper walk inside or outside.\ifconf\else\footnote{We recall that 
  nontrivial walks are defined in Sect.~\ref{s:ht_drawings}, a bit below
Corollary~\ref{c:projective_drawings}.}
\fi
 For this purpose, we create two new graphs $A^+$ and
$A^-$, possibly with loops but without multiple edges. 
In order to distinguish
these graphs from $G$, we draw their edges with double arrows and
we call these graphs an \emph{inside arrow graph} and an \emph{outside arrow
graph}, respectively. The edges of these graphs are called 
the \emph{inside/outside arrows}. We set $V(A^+) = V(A^-) = V(Z)$.

Now we describe the \emph{arrows}, that is, $E(A^+)$ and $E(A^-)$.
Let $u$ and $v$ be two vertices of $V(Z)$, not necessarily distinct. 
By $W^+_{uv}$ we denote the set of all proper nontrivial walks in $\gin$ with endpoints
$u$ and $v$.
We have an \emph{inside arrow} connecting $u$ and $v$ in $E(A^+)$ if and only if $W_{uv}^+$ is nonempty.
In order to distinguish the edges of $G$ from the arrows, we denote an arrow by
$\ar uv = \ar vu$. An arrow which is a loop at a vertex $v$ is denoted by $\ar vv$. (This
convention will allow us to work with arrows $\ar uv$ without a
distinction whether $u = v$ or $u \neq v$.)
Analogously, we define the set $W^-_{uv}$ and the \emph{outside arrows}.
\ifconf
\else

See Fig.~\ref{f:K_5_arrows} for the arrow graph(s) of the drawing of $K_5$ depicted in
Fig.~\ref{f:K_5_example}, left.

It follows from the definition of the inside bridges that any walk $\omega \in
W^+_{uv}$ stays in one inside bridge. Given an inside bridge $B$, we let
$W^+_{uv,B}$ be the set of all walks $w \in W^+_{uv}$ which belong to $B$.
In particular, $W^+_{uv}$ decomposes into the disjoint union of the sets
$W^+_{uv,B_1}, \dots, W^+_{uv,B_k}$ where $B_1, \dots, B_k$ are all inside
bridges.  
Given an inside arrow $\ar uv$ and an inside bridge $B$, we say that
$B$ \emph{induces} $\ar uv$ if
$W^+_{uv,B}$ is nonempty. 
An inside bridge $B$ is \emph{nontrivial} if it induces at least one arrow. 
Given two inside arrows $\ar uv$ and $\ar xy$
we say that $\ar uv$ and $\ar xy$ \emph{are induced by different bridges}
if there are two different inside bridges $B$ and $B'$ such that $B$ induces
$\ar uv$ and $B'$ induces $\ar xy$. As usual, we define analogous notions for
the outside as well. Note that it may happen that there is an
inside bridge inducing both $\ar uv$ and $\ar xy$ even if $\ar uv$ and $\ar
xy$ are induced by different bridges.

\begin{figure}
\begin{center}
\includegraphics{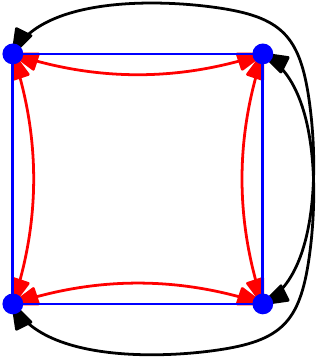}
\caption{The inside and the outside arrows corresponding to the drawing of $K_5$ from
Fig.~\ref{f:K_5_example}, left.}
\label{f:K_5_arrows}
\end{center}
\end{figure}

\fi

\paragraph{Possible configurations of arrows.}

\ifconf
\else
We plan to utilize the arrow graph in the following way. On one hand, we will
show that certain configurations of arrows are not possible; see
Fig.~\ref{f:forbidden_arrows}. On the other hand, we will show that, since the
arrow graph does not contain any of the forbidden configurations, it must
contain one of the configurations in Fig.~\ref{f:redrawable} inside or
outside. (These configurations are precisely defined in
Definition~\ref{d:redrawable}.)
We will also show that the configurations in Fig.~\ref{f:redrawable} are
\emph{redrawable}, that is, they may be appropriately redrawn without the
crosscap. The precise statement for redrawings is given by
Proposition~\ref{p:redrawings} below.

\begin{figure}
\begin{center}
\includegraphics{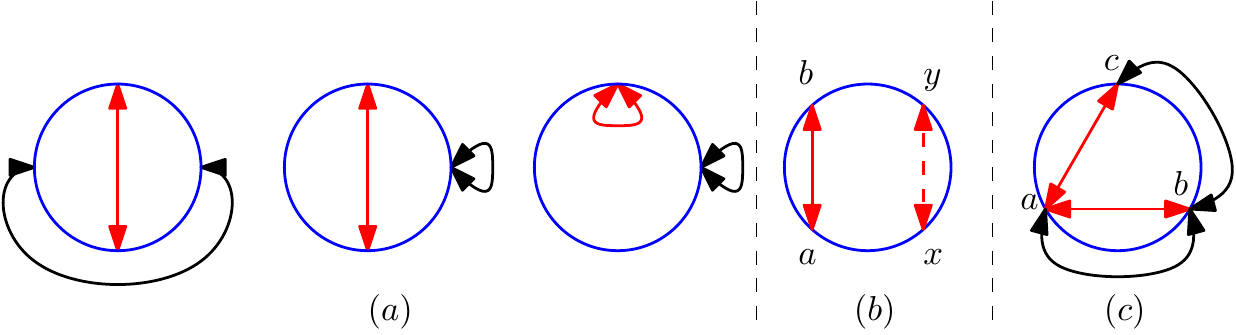}
\caption{Forbidden configurations of arrows. The cyclic order in $(a)$ may be
  arbitrary whereas it is important in $(b)$ that the arrows there do not
  interleave. Different dashing of lines in $(b)$
correspond to arrows induced by different inside bridges. The arrows of the same colour in $(c)$ are induced by the same bridge.}
\label{f:forbidden_arrows}
\end{center}
\end{figure}

More concretely, we prove the following three lemmas forbidding the
configurations of arrows from Fig.~\ref{f:forbidden_arrows}. We emphasize
that in all three lemmas we assume that the notions used there correspond to a fixed $G$,
$(D,\lambda)$ and $Z$ satisfying the separation assumptions.

\begin{lemma}\label{lem:arrows_touch}
Every inside arrow shares a vertex with every outside arrow.
\end{lemma}

\begin{lemma}\label{lem:disjoint_arrows_interleave}
Let $\ar{a}{b}$ and $\ar{x}{y}$ be two arrows induced by different inside
bridges of $\gin$. If the two arrows do not share an endpoint, their
endpoints have to interleave along $Z$. 
\end{lemma}

\begin{lemma}
\label{l:no_triangle_arrows}
There are no three vertices $a$, $b$, $c$ on $Z$, an inside bridge $B^+$,
and an outside bridge $B^-$ such that $B^+$ induces
the arrows $\ar ab$ and $\ar ac$ (and no other arrows) and $B^-$ induces 
the arrows $\ar ab$ and $\ar bc$ (and no other arrows).
\end{lemma}

We prove these three lemmas in
Sect.~\ref{s:forbidden_arrows}.
By symmetry, 
Lemmas~\ref{lem:disjoint_arrows_interleave} and~\ref{l:no_triangle_arrows} are
also valid if we swap the inside and the outside (Lemma~\ref{lem:arrows_touch} as well,
but here already the statement of the lemma is symmetric).
\bigskip

Now we describe the redrawable configurations.
\fi
\begin{definition}
  \label{d:redrawable}
  We say that $G$ forms
  \begin{enumerate}[$(a)$]
\item
  an \emph{inside fan} if there is a vertex
  common to all inside arrows. (The arrows may come from various inside
  bridges.)
\item
  an \emph{inside square} if it contains four vertices $a$, $b$, $c$ and $d$ ordered in this cyclic
  order along $Z$ and the inside arrows are precisely $\ar ab$, $\ar bc$, $\ar cd$ and $\ar
  ad$. In addition, we require that the inside graph $\gin$ has only one nontrivial inside bridge.
\item 
 an \emph{inside split triangle} if there exist three vertices $a$, $b$ and
  $c$ such that the arrows of~$G$ are $\ar ab$, $\ar ac$ and $\ar
  bc$. In addition, we require
  that every nontrivial inside bridge induces either the two arrows $\ar ab$ and $\ar
  ac$, or just a single arrow.
\end{enumerate} \ifconf \else
See Fig.~\ref{f:redrawable}. \fi We have analogous definitions for an
\emph{outside fan}, \emph{outside square} and
\emph{outside split triangle}.
\end{definition}

More precisely the notions in Definition~\ref{d:redrawable} depend on $G$,
$(D,\lambda)$ and $Z$ satisfying the separation assumptions.

\ifconf
\else

\begin{figure}
\begin{center}
\includegraphics{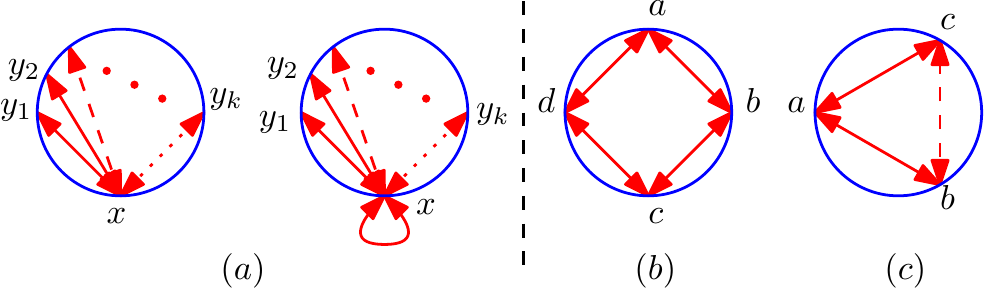}
\caption{Schematic drawings of the redrawable configurations of arrows from Definition~\ref{d:redrawable}.
  Different dashing of lines correspond to different inside bridges. The
  loop in the right drawing $(a)$ is an inside loop (drawn outside due to lack of space).
  The drawing $(c)$ is only one instance of an inside split triangle.
} 
\label{f:redrawable}
\end{center}
\end{figure}

\fi

A relatively direct case analysis, using \ifconf\else
Lemmas~\ref{lem:arrows_touch},~\ref{lem:disjoint_arrows_interleave}
and~\ref{l:no_triangle_arrows}, \fi reveals the following fact.

\begin{proposition}
\label{p:redrawable_exist}
Let $(D,\lambda)$ be a projective HT-drawing on $S^2$ of a graph $G$ and let $Z$ be 
a cycle in $G$ satisfying the separation assumptions.
Then $G$ forms an (inside or outside) fan, square, or split triangle.

\end{proposition}

On the other hand, any configuration from Definition~\ref{d:redrawable}
can be redrawn without using the crosscap:
\begin{proposition}
\label{p:redrawings}
Let $(D, \lambda)$ be a projective HT-drawing of $\gin$ on $S^2$ and $Z$ be a
cycle satisfying the separation assumptions. 
Moreover, let us assume that $D(\gin) \cap S^-
= \emptyset$ (that is, $\gin$ is fully drawn on $S^+ \cup D(Z)$). Let us also assume that $\gin$ forms an inside fan, 
an inside square or an inside split triangle. Then there is an
ordinary HT-drawing $D'$ of $\gin$ on $S^2$ such that $D$ coincides with $D'$
on $Z$ and $D'(\gin) \cap S^-
= \emptyset$.
\end{proposition}

\ifconf
\else
Proposition~\ref{p:redrawable_exist} is proved in
Sect.~\ref{s:labellings}
(assuming there the validity of Lemmas~\ref{lem:arrows_touch},~\ref{lem:disjoint_arrows_interleave}
and~\ref{l:no_triangle_arrows}).
Proposition~\ref{p:redrawings} is proved in
Sect.~\ref{s:redrawings}.
\fi

Now we are missing only one tool to finish the proof of
Theorem~\ref{t:black_box}. This tool is the ``redrawing procedure'' of
Pelsmajer, Schaefer and \v{S}tefankovi\v{c}~\cite{PSS07}. More concretely, we
need the following variant of Theorem~\ref{t:2.1}.
(Note that the theorem below is not in the setting of projective HT-drawings.
However, the notions used in the statement are still well defined according to 
Definition~\ref{d:in_out}.)

\begin{theorem}
\label{t:PSS_in_out}
Let $D$ be a drawing of a graph $G$ on the sphere $S^2$. Let $Z$ be a cycle in
$G$ such that every edge of $Z$ is even and $Z$ is drawn as a simple cycle.
Then there is a drawing $D''$ of $G$ such that
\begin{itemize}
  \item $D''$ coincides with $D$ on $Z$;
  
  \item $D''(\gin)$ belongs to $S^+ \cup D(Z)$ and $D''(\gout)$ belongs to $S^-
    \cup D(Z)$; 
  \item
    whenever $(e,f)$ is a pair of edges such that both $e$ and $f$ are inside
    edges or both $e$ and $f$ are outside edges, then $\crno_{D''}(e,f) =
    \crno_{D}(e,f)$.
\end{itemize}
\end{theorem}

It is easy to check that the proof of Theorem~\ref{t:2.1}
in~\cite{PSS07} proves
Theorem~\ref{t:PSS_in_out} as well. Additionally, we note that an alternative proof of Theorem~\ref{t:2.1} in~\cite[Lemma~3]{Fulek2012AdjCrossings} can also be extended to yield Theorem~\ref{t:PSS_in_out}. \ifconf\else Nevertheless, for \fi
completeness, we provide its proof
in
\ifconf
\else
Sect.~\ref{s:PSS_in_out}. 
\fi

Finally, we prove
Theorem~\ref{t:black_box}, assuming the validity of the aforementioned auxiliary results.

\ifconf
\else
\begin{proof}[Proof of Theorem~\ref{t:black_box}]
  Let $G$ be the graph, $(D, \lambda)$ be the drawing and $Z$ be the cycle from
  the statement.

We use Theorem~\ref{t:PSS_in_out} to $G$ and $D$ to obtain a drawing $D''$
keeping in mind that all edges of $Z$ are even. See
Fig.~\ref{f:K_5_projective}; follow this picture also in the next steps of the
proof. We get that $Z$ is drawn on $D''$ 
as a simple cycle free of crossings. We also get that $D''(\gin)$ is contained in 
$\inn{S} \cup D''(Z)$ and $D''(\gout)$ is contained in $\out{S} \cup D''(Z)$.
 However, there may be no $\lambda''$
such that $(D'', \lambda'')$ is a projective HT-drawing; we still may need to modify it
to obtain such a drawing.

By Proposition~\ref{p:redrawable_exist}, $G$ forms one of the
redrawable configurations on one of the sides; that is, an inside/outside fan,
square or split triangle. Without loss of generality, it appears inside. It
means that $D''$ restricted to $\gin$ satisfies the assumptions of
Proposition~\ref{p:redrawings}. Therefore, there is an ordinary HT-drawing
$D^+$ of $\gin$ satisfying the conclusions of Proposition~\ref{p:redrawings}.
Finally, we let $D'$ be the drawing of $G$ on $S^2$ which coincides with $D^+$
on $\gin$ and with $D''$ on $\gout$. Both $D''$ and $D^+$ coincide with $D$ on
$Z$; therefore, $D'$ is well defined. We set $\lambda'$ so that $\lambda'(e) :=
\lambda(e)$ for an edge $e \in E^-$ and $\lambda'(e) := 0$ for any other edge.
Now, we can easily verify that $(D',\lambda')$ is the required projective
HT-drawing.

\begin{figure}
\begin{center}
\includegraphics{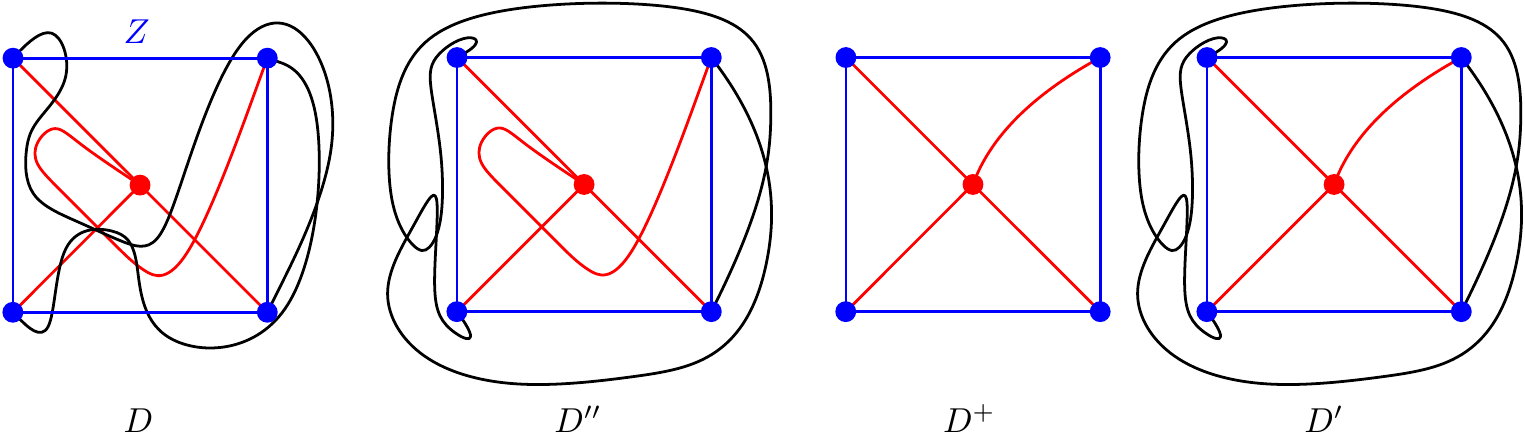}
\caption{Redrawing a projective HT-drawing of $K_5$ analogously to the drawing in
Fig.~\ref{f:K_5_example}.}
\label{f:K_5_projective}
\end{center}
\end{figure}

Indeed, let $e$ and $f$ be independent edges. If both $e$ and $f$ are inside
edges, then $\crno_{D'}(e,f) = \crno_{D^+}(e,f) = 0 = \lambda'(e)\lambda'(f)$,
since $D^+$ is an ordinary HT-drawing. If both $e$ and $f$ are outside edges,
then $\crno_{D'}(e,f) = \crno_{D''}(e,f) = \crno_{D}(e,f) =
\lambda(e)\lambda(f) = \lambda'(e)\lambda'(f)$. Finally, if one of this edges
is an inside edge and the other is an outside edge, then $\crno_{D'}(e,f) = 0 =
\lambda'(e)\lambda'(f)$, because $D'(e)$ and $D'(f)$ are separated by $D'(Z)$.
\end{proof}
\fi

\section[\texorpdfstring{Proof of the Strong Hanani--Tutte Theorem on $\R P^2$}{Proof of the Strong Hanani--Tutte Theorem on the Projective Plane}]{Proof of the Strong Hanani--Tutte Theorem on $\R P^2$}
\label{s:induction}

\ifconf
\else
In this section, we prove Theorem~\ref{thm:main} assuming validity of
Theorem~\ref{t:black_box} as well as few other auxiliary results from the
previous section,
which will be proved only in the later sections. 
\fi

Given a graph $G$
that admits an HT-drawing on the projective plane, we need to show that $G$ is
actually projective-planar. By \ifconf
\else
Corollary~\ref{c:projective_drawings}, \fi
we may assume that $G$ admits a projective $HT$-drawing $(D,\lambda)$ on
$S^2$. We head for using Theorem~\ref{t:black_box}. For this, we need that
$G$ is $2$-connected and contains a suitable trivial cycle $Z$ that may be 
redrawn so that it satisfies the assumptions of Theorem~\ref{t:black_box}.
Therefore, we start with auxiliary claims that will bring us to this setting.
Many of them are similar to auxiliary steps
in~\cite{PSS07} (sometimes they are almost identical, adapted to a new setting).

Before we state the next lemma, we recall the well known fact that any graph
admits a (unique) decomposition into \emph{blocks of 2-connectivity}~\cite[Ch.~3]{diestel10}. Here, we also allow
the case that $G$ is disconnected.
\ifconf \else
Each block in this decomposition is
either a vertex (this happens only if it is an isolated vertex of $G$), an edge
or a 2-connected graph with at least three vertices. The intersection of two
blocks is either empty or it contains a single vertex (which is a cut in the
graph). The blocks of the decomposition
cover all vertices and edges (a vertex may occur in several blocks whereas any
edge belongs to a unique block).
\fi

\begin{lemma}\label{l:2-connected}
If $G$ admits a projective HT-drawing on $S^2$,
then at most one block of 2-connectivity in $G$ is non-planar.
Moreover, if all blocks are planar, $G$ is planar as well.
\end{lemma}

\ifconf
\else
We note that in~\cite{SchaeferZ2_embed} it was proved that a minimal counterexample to the strong Hanani--Tutte theorem on any surface is vertex 2-connected. However, for the projective plane the same result can be obtained by much simpler means; therefore, we include its proof here.
\begin{proof}
First, for contradiction, let us assume that $G$ contains two distinct
non-planar blocks $B_1$ and $B_2$. If $B_1$ and $B_2$ are disjoint, then
Lemma~\ref{l:disjoint_cycles_projective} implies that at least one of these
blocks, say $B_2$, does not contain any non-trivial cycle. However, it means
that $B_2$ admits an ordinary HT-drawing on $S^2$ by
Lemma~\ref{l:planarize}. Therefore, $B_2$ is planar by the strong Hanani--Tutte
theorem in the plane~\cite{Hanani34,Tutte70,PSS07}. This contradicts our original assumption.
It remains to consider the case when $B_1$ and $B_2$ share a vertex $v$ (it
must be a cut vertex).
 Let us set $H:=B_1\cup B_2$. Let $P$ be a spanning tree of $H$ with just two edges $e_1, e_2$ incident to $v$ 
 and such that $e_1 \in B_1$ and $e_2 \in B_2$.
 Note that such a tree always exists, because $B_1$ and $B_2$ are connected after removing $v$.
 By Lemma~\ref{l:planarize} we may assume that all the edges of $P$ are
 trivial (after a possible alteration of $\lambda$).
 
 Any nontrivial edge $e$ from $E(H) \setminus E(P)$ creates a nontrivial cycle
 in the corresponding block. If $e$ is not incident to $v$, then the cycle
 avoids $v$ by the choice of $P$.
 Using Lemma~\ref{l:disjoint_cycles_projective} again, we see that at least one
 of the blocks, say $B_2$, satisfies that all its nontrivial edges are incident
 with $v$.   
   This already implies that $B_2$ is a planar graph, because $D$ is an
 $HT$-drawing of $B_2$ on $S^2$ (there are no pairs of nontrivial independent
 edges in $G$). This is again a contradiction.

The last item in the statement of this lemma is a well known property of planar
graphs. It is sufficient to observe that a disjoint union of two planar graphs
is a planar graph, and moreover, that if a graph $G$ contains a cut vertex
$v$ and all the components after cutting (and reattaching $v$) are planar, then
$G$ is planar as well.
\ifconf
\qed
\fi
\end{proof}
\fi

\begin{observation}\label{obs:no_trivial_cycle}
Let $(D,\lambda)$ be a drawing of a $2$-connected graph.
If  $D$ does not contain any trivial cycle, then $G$ is planar.
\end{observation}

\ifconf\else
\begin{proof}
As $G$ is $2$-connected, it is either a cycle or it contains three disjoint
paths sharing their endpoints. A cycle is a planar graph as we need. In the
latter case, two of the paths are both trivial or both nontrivial. Together,
they induce a trivial cycle, therefore this case cannot occur.
\ifconf
\qed
\fi
\end{proof}
\fi

\begin{lemma}\label{lem:even}
  Let $(D, \lambda)$ be a projective HT-drawing on $S^2$ of a graph $G$ and let $Z$ be a
cycle in $G$. Then $G$ can be redrawn only by local changes
next to the vertices of $Z$ to a projective HT-drawing $D'$ on $S^2$
so that $\lambda$ remains unchanged and $cr_{D'}(e,f)=\lambda(e)\lambda(f)$, 
for any pair $(e,f)\in E(Z)\times E(G)$ of distinct (not necessarily
independent) edges.
In particular, if $\lambda(e)=0$ for every edge $e$ of $Z$, then every edge of $Z$ becomes even in $D'$.
\end{lemma}

\ifconf\else
\begin{proof}
Since we have a projective HT-drawing, $cr_D(e,f)=\lambda(e)\lambda(f)$
for every pair of independent edges.
To prove the claim it remains to show that local changes allow
to change the parity of $cr_D(e,f)$ whenever $e$ is an edge of $Z$ and $e$ and $f$ share a vertex.

This can be done in two steps. First we use local move c) from Fig.~\ref{f:local_changes} to obtain
the desired parity of $cr_D(e,f)$, for all pairs of consecutive edges $(e,f)$ on $Z$. 
This move may change the parity of crossings between edges on $Z$ and dependent edges not on $Z$.

Next we use local moves a) and b) from Fig.~\ref{f:local_changes} to obtain the desired parity 
of crossings between edges on $Z$ and dependent edges not on $Z$. 
If $v$ is the vertex common to $h$, $e$ and $f$, where $e$ and $f$ are edges on $Z$,
move a) is used when we need to change
the parity of $cr_D(e,h)$ and its symmetric version to change the parity of $cr_D(f,h)$.
Move b) is used when we need to change the parity for both $cr_D(e,h)$ and $cr_D(f,h)$.
Since these moves do not change the parity of $cr_D(e,h')$ or $cr_D(f,h')$ for any other edge $h'$,
the claim follows.
\begin{figure}
\begin{center}
\includegraphics{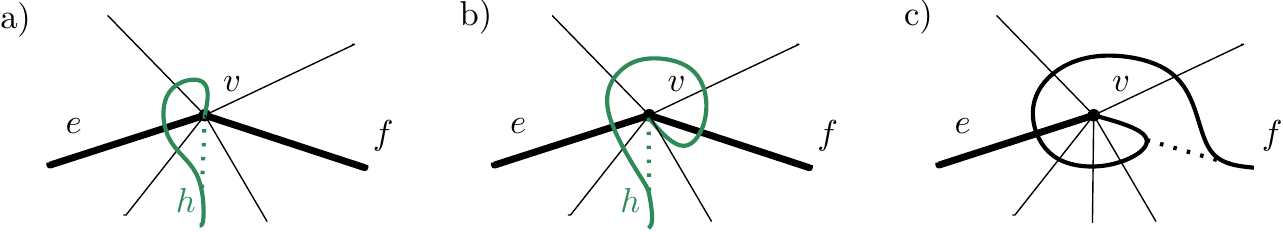}
\caption{Local changes to make all edges of $Z$ even. The original drawing of the edge near $v$ is dotted.}
\label{f:local_changes}
\end{center}
\end{figure}
\ifconf
\qed
\fi
\end{proof}
\fi

Once we know that the edges of a cycle can be made even we also need to know that such a cycle can be made simple.

\begin{lemma}\label{lem:simple}
  Let $(D, \lambda)$ be a projective HT-drawing on $S^2$ of a graph $G$ and
  let $Z$ be a cycle in $G$ such that each of its edges is even. Then $G$ can be redrawn so that $Z$ becomes a simple cycle, its edges remain even and the resulting drawing is still a projective HT-drawing (with $\lambda$ unchanged). 
\end{lemma}

\ifconf\else
\begin{proof}
First, we want to get a drawing such that there is only one edge of $Z$
which may be intersected by other edges.  Let us consider three consecutive
vertices $u$, $v$ and~$w$ on~$Z$, with $v\not\in\{u,w\}$.  We
\emph{almost-contract} $uv$ so that we move the vertex $v$ towards $u$ until we
remove all intersections between $uv$ and other edges.  Note that the image of
the cycle~$Z$ is not changed; we only slide $v$ towards~$u$ along~$Z$.
This way, $uv$ is now free of
crossings and these crossings appear on $vw$. 
See the two leftmost pictures in Fig.~\ref{f:almost_contract}. (The right
picture will be used in the proof of Theorem~\ref{t:PSS_in_out}.)

Since $uv$ as well as $vw$ were even edges in the initial drawing, $vw$ remains even after the redrawing.
If $uv$ and $vw$ intersected, then this step introduces self-intersections of $vw$.
\begin{figure}
\begin{center}
\includegraphics{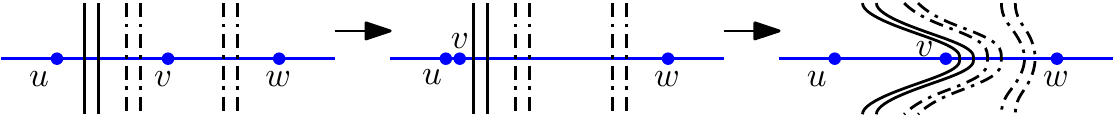}
\caption{Almost contracting an edge.}
\label{f:almost_contract}
\end{center}
\end{figure}

After performing such redrawing repeatedly, we get there is only one edge of
$Z$ which may be intersected by other edges, as required. We remove
self-crossings of this edge, as described in Sect.~\ref{s:ht_drawings}, and we are done. 
\end{proof}
\fi

\ifconf
\else
Apart from lemmas tailored to set up the separation assumptions, we also
need one more lemma that will be useful in the inductive proof of
Theorem~\ref{thm:main}.

\begin{lemma}
\label{l:single_arrow}
  Let $(D, \lambda)$ be a Hanani--Tutte drawing of $G$ and let $Z$ be a cycle
  satisfying the separation assumptions. Let $B$ be an inside bridge such that
  any path with both endpoints on $V(B) \cap V(Z)$ is nontrivial. Then $|V(B)
  \cap V(Z)| = 2$ and $B$ induces a single arrow and no loop.
\end{lemma}

\begin{proof}
First, we show that there is no nontrivial cycle in $B$. For contradiction, 
there is a nontrivial cycle $N$ in $B$.
By the $2$-connectivity of $G$ there exist two vertex disjoint paths $p_1$ and
$p_2$ (possibly of length zero) that connect $Z$ to $N$. We consider the shortest
such paths; thus, each of the paths shares only one vertex with
$Z$ and one vertex with $N$. Let $y_1$ and $y_2$ be the endpoints of $p_1$ and
$p_2$ on $N$, respectively. 
Let $p_3$, $p_4$ be the two arcs of $N$ between $y_1$ and $y_2$. We consider
two paths $q_1$ and $q_2$ where $q_1$ is obtained from the concatenation of
$p_1$, $p_3$ and $p_2$, while $q_2$ is obtained from the concatenation of $p_1$,
$p_4$ and $p_2$. Since $N$ is non-trivial, one of these paths is trivial, which
provides the required contradiction.

Next, we observe that $B$ does not induce any loop in the inside arrow graph.
For contradiction, it induces a loop at a vertex $x$ of $Z$. This means that
there is a proper nontrivial walk $\kappa$ in $B$ with both endpoints $x$. We set up
$\kappa$ so that it is the shortest such walk. We already know that $\kappa$
cannot be a cycle, thus it contains a closed nonempty subwalk $\kappa'$ and we
set up $\kappa'$ so that it is the shortest such subwalk. Therefore, it must be
a cycle; by the previous part of this proof, it is trivial. However, it
means that $\kappa$ can be shortened by leaving out $\kappa'$, which is the
required contradiction.
 
Now, we show that $|V(B) \cap V(Z)| = 2$. By the $2$-connectedness of $G$, we have
that  $|V(B) \cap V(Z)| \geq 2$. Thus, for contradiction,
let $a,b,c$ be three distinct vertices of $V(B) \cap V(Z)$. Let $v$ be one of
the inner vertices of $B$ (there must be such a vertex since $B$ cannot be a
single edge in this case). By the definition of inside/outside bridges, there exist proper walks $p_a$, $p_b$ and $p_c$
 connecting $v$ to $a,b$ and $c$, respectively.
 By the pigeonhole principle, two of the walks have the same value of $\lambda$;
 without loss of generality, let them be $p_a$ and $p_b$.
 It follows that the proper walk obtained from the concatenation of $p_a$ and
 $p_b$ is trivial. Since $B$ does not contain any non-trivial cycle, this walk
 can be shortened to a trivial proper path between $a$ and $b$ by an analogous
 argument as in the previous paragraph. A contradiction.

 Finally, we know that there are two vertices in $V(B) \cap V(Z)$. Let $x$ and
 $y$ be these two vertices. Since any path connecting $x$ and $y$ is
 nontrivial, $B$ induces the arrow $\ar xy$ in $A^+$. No other arrow in $A^+$ is
 possible since there are no loops.
\ifconf
\qed
\fi
\end{proof}
\fi

Proposition~\ref{p:inductive_hypothesis} below is our main tool for deriving Theorem~\ref{thm:main} from
Theorem~\ref{t:black_box}. It is set up in such a way that it can be
inductively proved from Theorem~\ref{t:black_box}. Then it
implies Theorem~\ref{thm:main}, using the auxiliary lemmas from the beginning
of this section, relatively easily.

\begin{proposition}\label{p:inductive_hypothesis}
Let $(D,\lambda)$ be a projective HT-drawing of a $2$-connected graph $G$ on $S^2$ and 
$Z$ a cycle in~$G$ that is completely free of crossings in~$D$ and such that each of its edges is trivial in $D$. 
Assume that  $\zin$ or $\zout$ is empty.  Then $G$ can be embedded into $\RP^2$
so that $Z$ bounds a face of the resulting \ifconf
\else embedding homeomorphic to a disk. \fi If, in addition, $D$ is an ordinary HT-drawing on $S^2$, then $G$ can be embedded into $S^2$ so 
that $Z$ bounds a face of the resulting \ifconf\else embedding (this face is again homeomorphic to a
disk---there is in fact no other option on $S^2$)\fi.\footnote{We need to consider
the case of ordinary HT-drawings in this proposition for a well working
induction.}
\end{proposition}

\ifconf
\begin{proofsketch}[of Proposition~\ref{p:inductive_hypothesis}]
\else
\begin{proof}
\fi
The proof proceeds by induction on the number of edges of $G$.  
The base case is when $G$ is a cycle.

Without loss of generality, we assume that $\zout$ is empty. 
That is, $G=\gin$. If $\zin$ is also empty,
$G$ consists only of $Z$ and such a graph can easily be embedded into the plane
or projective plane as required. Therefore, we assume that $\zin$ is nonempty.

We find a path $\gamma$ in $(V(\gin), E(\gin)\setminus E(Z))$ connecting two points $x$ and $y$ lying on $Z$. 
We may choose $x,y$ so that $x\neq y$ since $G$ is $2$-connected.

\paragraph[\texorpdfstring{Case 1: There exists a trivial $\gamma$.}{Case 1: There exists a trivial gamma.}]{Case 1: There exists a trivial $\gamma$.}
\ifconf
\else
First we solve the case that at least one such path $\gamma$ is trivial.
We show that all edges of $\gamma$ can be made even and simple in the drawing while 
preserving simplicity of $Z$, the fact that $Z$ is free of crossings and the projective Hanani--Tutte 
condition on the whole drawing of $\gin$. 

As the first step, we use Lemma~\ref{l:planarize} in order to achieve that
$\lambda(e) = 0$ for any edge $e$ of $Z$ and $\gamma$ simultaneously. By
inspecting the proof of Lemma~\ref{l:planarize} we see that we can achieve this
by vertex-crosscap switches only over the inner vertices of $\gamma$ (for this,
we set up the root in the proof to be one of the endpoints of $\gamma$). In
particular we can perform these vertex-crosscap switches inside $Z$ without
affecting $Z$.

Now, we want to make the edges of $\gamma$ even, again without affecting $Z$.
First, for any pair $(e,f)$ of adjacent edges of $\gamma$ which intersect
oddly, we locally perform the move c) from Fig.~\ref{f:local_changes} similarly as in 
Lemma~\ref{lem:even}. Next, we consider any edge $e\notin E(\gamma)$ adjacent to a vertex $u\in V(\gamma)\setminus V(Z)$. 
For such an edge we eventually perform one of the moves a) or b) from
Fig.~\ref{f:local_changes} so that we achieve that $e$ intersects evenly each
of the two edges of $\gamma$ incident with $u$. Finally, we consider any edge
$e\notin E(\gamma)\cup E(Z)$ adjacent to $u\in\{x,y\}$, one of the endpoints of $\gamma$ on $Z$.
Let $f$ be the edge of $\gamma$ incident with $u$. If $e$ and $f$ intersect oddly, we
perform the move from Fig.~\ref{f:local_changes2}. This is possible since $Z$
is free of crossings. This way we achieve that every edge of $\gamma$ is even.

\begin{figure}
\begin{center}
  \includegraphics{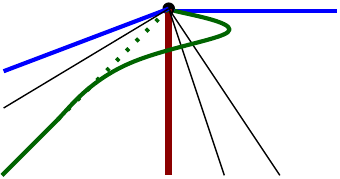}
  \caption{Local changes at $u$. The original drawing of the edge is dotted, $Z$ is depicted in blue, $\gamma$ in red.
  The changed edge in green.}
  \label{f:local_changes2}
\end{center}
\end{figure}

As the last step of the redrawing of $\gamma$, we want to make $\gamma$ simple (again
without affecting $Z$). This can be done in the same way as in
Lemma~\ref{lem:simple}. We almost-contract all edges of $\gamma$ but one so
that there is only one edge of $\gamma$ that intersects with other edges. Then
we remove eventual self-intersections.

The rest of the argument is easier to explain if we switch inside and outside (this is easily doable by a homeomorphism of
$S^2$) and treat drawings on $S^2$ as drawings in the plane. 

We may assume that after the homeomorphism $Z$ is drawn in the plane as
a circle with the inner region empty and with $x$ and $y$ antipodal. 
The vertices $x$ and $y$ split $Z$ into two paths; we denote by $p_1$ the
`upper' one and by $p_2$ the `lower' one.  We may also assume that $\gamma$ 
is `above' $p_1$ by eventually adapting the initial choice of the
correspondence between $S^2$ and the plane.

Now we continuously deform the plane so that $Z$ becomes flatter and flatter
until it coincides with the line segment connecting $x$ to $y$, as depicted in
Fig.~\ref{f:flattening} a).
We may further require that no inner vertex of $p_1$ was identified with any inner vertex of $p_2$.
\begin{figure}
\begin{center}
  \includegraphics{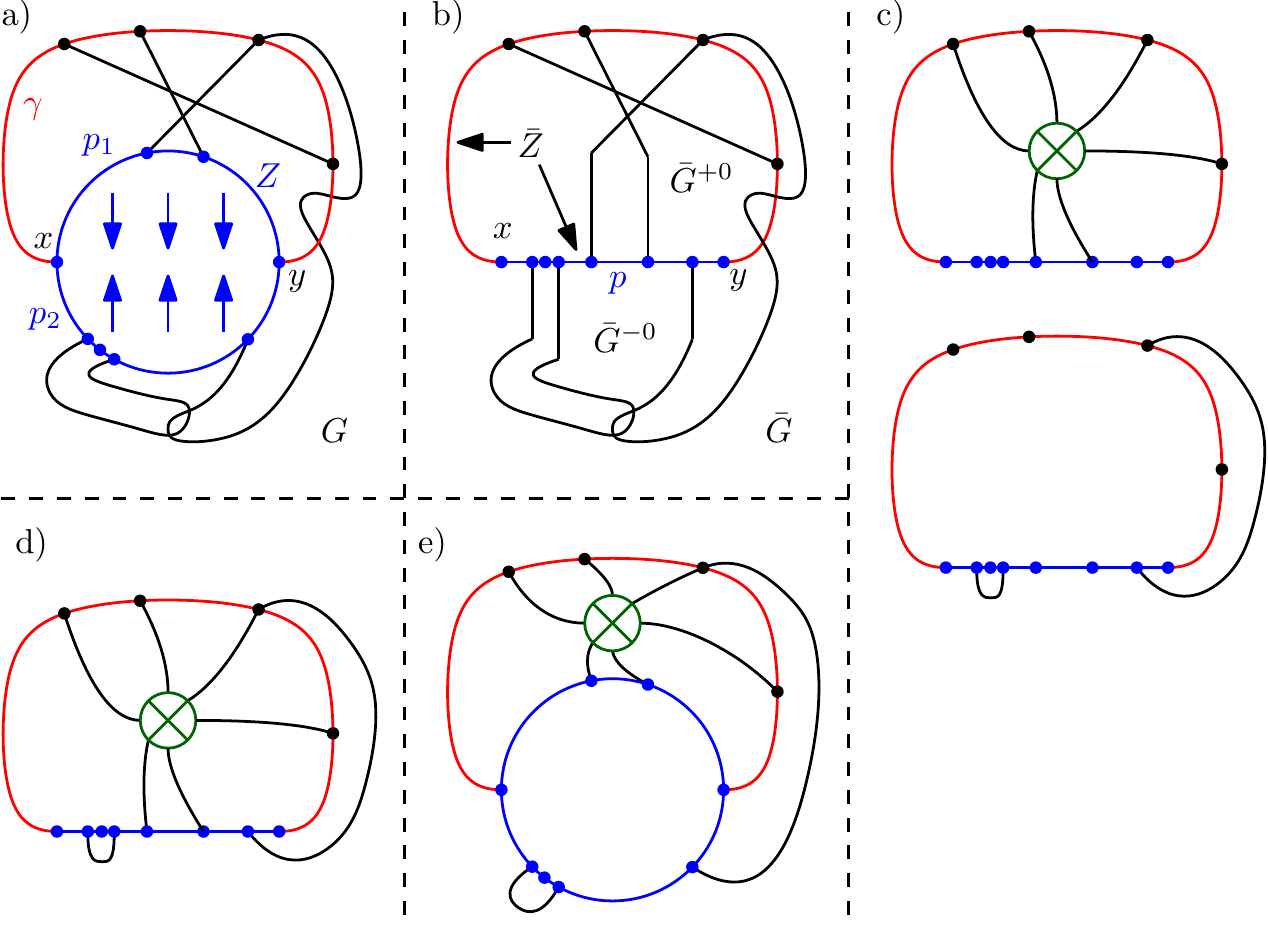}
  \caption{The deformation of the plane that changes $G$ into $\bar{G}$, the redrawing of $\bar{G}$ and the resulting embeddings of $\bar G$ and $G$.}
  \label{f:flattening}
\end{center}
\end{figure}

This way, we get a projective HT-drawing $(\bar D,\bar \lambda)$ of a new graph
$\bar G$: all the vertices
of $G$ remain present in $\bar G$, that is, $V(G) = V(\bar G)$. Also the edges of $G$ which are not on $Z$
are present in $\bar G$. Only some of the edges of $Z$ may disappear and they are
replaced with edges forming a path $p$ between $x$ and $y$. Note that we
did not introduce any multiple edges, because there is no edge in $G$
connecting an inner vertex of $p_1$ with an inner vertex of $p_2$. It also
turns out that $\bar{G}$ has one edge less than $G$. Regarding $\bar \lambda$, we have
$\lambda(e) = \bar \lambda(e)$ if $e$ is an edge of $E(G) \setminus E(Z)$ and we
have $\bar \lambda(e) = 0$ if $e$ belongs to $p$.

Now consider the cycle $\bar Z$ in $\bar G$ formed by $\gamma$ and $p$. It is trivial and
simple. In particular, we distinguish the inside and the outside according to
Definition~\ref{d:in_out}. For example, $\bar G^{+0}$ corresponds to the part
of $G$ in between $\gamma$ and $p_1$ before the flattening; see
Fig.~\ref{f:flattening} a) and b). 

Now, we apply Theorem~\ref{t:black_box} and we get 
a drawing $D'$ of $\bar G$. When we look at the two sides of $\bar G$
separately, we get that the drawing of one of the sides, say the
drawing of $\bar G^{+0}$, is a projective HT-drawing, while there is an ordinary
HT-drawing on $S^2$ on the other side. If, in addition, $D$ were already an
ordinary HT-drawing, we get an ordinary HT-drawing on both sides by
Theorem~\ref{t:PSS_in_out}.

Note also that since $G$ was $2$-connected, both parts of
$\bar G$ are $2$-connected as well. Subsequently, we examine each of these two
parts separately and use the inductive hypothesis; we obtain an embedding of
$\bar G^{+0}$ into $\RP^2$ such that $\bar Z$ bounds a face homeomorphic to a
disk as well as an embedding of $\bar G^{-0}$ into $S^2$ such that
$\bar Z$ bounds a face homeomorphic to a disk. If, in addition, $D$ were already
an ordinary HT-drawing, we get also the required embedding of $\bar G^{+0}$
into $S^2$.
We merge these two embeddings
along $\bar Z$ obtaining an embedding of $\bar G$ into $\RP^2$ (or $S^2$ if
$D$ were an ordinary HT-drawing). See Fig.~\ref{f:flattening} c) and d).

Finally, we need to undo the identification of $p_1$ and $p_2$ into $p$.
Whenever we consider a vertex $v$ on $p$ different from $x$ and $y$, it is
uniquely determined whether it comes from $p_1$ or $p_2$. In addition, if $v$
comes from $p_1$, then any edge $e \in E(G) \setminus E(Z)$ incident with $v$
must belong to $\bar G^{+0}$. Similarly, if $v$ 
comes from $p_1$, then any edge $e \in E(G) \setminus E(Z)$ incident with $v$
must belong to $\bar G^{-0}$. Therefore, it is possible to undo the
identification and we get the required embedding of $G$. See
Fig.~\ref{f:flattening} e).

\fi

\paragraph[\texorpdfstring{Case 2: All choices of $\gamma$ are nontrivial.}{Case 2: All choices of gamma are nontrivial.}]{Case 2: All choices of $\gamma$ are nontrivial.}
\ifconf
Now, we need to resolve the case when all possible choices of $\gamma$ are
nontrivial. Let $A^{+0}$ be the graph obtained from the inside arrow graph $A^+$ 
by adding the edges of $Z$ (in particular, $Z$ is a subgraph of $A^{+0})$. 
We aim to show that $A^{+0}$ admits an embedding in $\RP^2$ such that $Z$ 
bounds a disk face. As soon as we show this, we aim to replace the embedding of 
each arrow of $A^{+0}$ by an embedding of the inside bridges inducing this
arrow (if there are more such bridges, we embed them in parallel). The key
fact that makes it possible is that each inside bridge meets
$Z$ in exactly two points and induces a single arrow and no loop. (Here, we
leave this fact without a proof.) We also need to check that each of the
bridges, together with $Z$, admits an embedding. This follows from
Proposition~\ref{p:redrawings} for inside fans and from Case 1 of this proof.

It remains to sketch why $A^{+0}$ admits the required embedding.
We know that any two disjoint arrows interleave using Lemma~\ref{l:impossible}(b).
Let us consider two concentric closed disks $E_1$ and $E_2$ such that $E_1$
belongs to the interior of $E_2$. Let us draw $Z$ to the boundary of $E_2$. Let
$a$ be the number of arrows of $A^+$ and let us consider $2a$ points on the
boundary of $E_1$ forming the vertices of a regular $2a$-gon. These points will be
marked by ordered pairs $xy$ where $\ar xy$ is an inside arrow. We mark 
the points so that the cyclic order of the points respects the cyclic order
on $Z$ in the first coordinate (the pairs with the same first
coordinate are consecutive). However, for a fixed $x$, the pairs $xy_1,
\dots, xy_k$ corresponding to all arrows emanating from $x$ are ordered in
the reverted order when compared with the order of $y_1, \dots, y_k$ on $Z$.

\begin{center}
\includegraphics{only_nontrivial_conference}
\end{center}

It is not hard to check that the points marked $xy$ and $yx$ are precisely
the opposite points.
Now, we get the required drawing in the following way. For any arrow $\ar xy$ we
connect $x$ with the point marked $xy$ and $y$ with $yx$. We can do
all the connections simultaneously for all arrows without introducing any crossing since we
have respected the cyclic order in the first coordinate. We remove the interior
of $E_1$ and identify the opposite points on its boundary. This way we introduce a crosscap.
Finally, we glue another disk along its boundary to $Z$ and we get the required
drawing on $\RP^2$.
\qed
\end{proofsketch}

\else

Now we deal with the situation when all possible choices of $\gamma$ are nontrivial.
We will first analyse which situations allow such configuration. Later we will show
how to draw each of these situations.

Let us consider the inside arrow graph $A^+$. 
Since all choices of $\gamma$ are nontrivial, Lemma~\ref{l:single_arrow} shows
that every inside bridge induces a single inside arrow. This allows us to
redraw inside bridges separately as is provided by the following claim.

\begin{claim}\label{claim:planar_bridge}
For any inside bridge $B$ there exists a planar drawing of $Z\cup B$
in which $Z$ is the outer face.
\end{claim}

\begin{proof}
 Since we know that $B$ induces only a single arrow, we get that 
 $Z\cup B$ forms an inside fan, according to
 Definition~\ref{d:redrawable}. It follows from Proposition~\ref{p:redrawings}
 that $Z\cup B$ admits an ordinary HT-drawing such that $Z$ is an outer cycle.
 However, the setting of ordinary HT-drawings is already fully resolved in Case~1.
That is, we may already use Proposition~\ref{p:inductive_hypothesis} for this
 drawing and we get the required conclusion.
\end{proof}

We consider the graph $A^{+0}$ obtained from $A^+$ by adding the
edges of $Z$ to it, where $A^+$ is the inside arrow graph. (Note that $V(A^+) = V(Z)$ according to our definition of
the arrow graph.)

Our main aim will be to find an embedding of $A^{+0}$ to $\RP^2$ such that $Z$
bounds a face. As soon as we reach this task, then we can replace an embedding
of each arrow by the embedding of inside bridges inducing this arrow via
Claim~\ref{claim:planar_bridge} in a close neighbourhood of the arrow. If there are, possibly, 
more inside bridges inducing the arrow,
then they are embedded in parallel. 

Finally, we show that it is possible to embed $A^{+0}$ in the required way. 
By Lemma~\ref{lem:disjoint_arrows_interleave}, any two disjoint arrows interleave.

Let us consider two concentric closed disks $E_1$ and $E_2$ such that $E_1$
belongs to the interior of $E_2$. Let us draw $Z$ to the boundary of $E_1$. Let
$a$ be the number of arrows of $A^+$ and let us consider $2a$ \emph{points} on the
boundary of $E_1$ making the vertices of regular $2a$-gon. These points will
marked by ordered pairs $(x, y)$ where $\ar xy$ is an inside arrow. We mark 
the points so that the cyclic order of the points respect the cyclic order as
on $Z$ in the first coordinate (in particular pairs with the same first
coordinate are consecutive). However, for a fixed $x$, the pairs $(x,y_1),
\dots (x,y_k)$ corresponding to all arrows emanating from $x$ are ordered in
the reverted order when compared with the order of $y_1, \dots, y_k$ on $Z$.
 See Fig.~\ref{f:only_nontrivial}.

We show that it follows that the points marked $(x, y)$ and $(y, x)$
are directly opposite on $E_1$ for every inside arrow $\ar xy$. For
contradiction, let us assume that $(x, y)$ and $(y, x)$ are not
directly opposite for some $\ar xy$. Then there is another arrow $\ar uv$ such
that $(x, y)$ and $(y, x)$ do not interleave with $(u, v)$ and
$(v, u)$. Indeed, such an arrow must exist because the arrows induce a
matching on the points, and $(x, y)$ and $(y, x)$ do not split the
points equally. However, if $\ar xy$ and $\ar uv$ do not share an endpoint, we
get a contradiction with the fact that disjoint arrows interleave. If $\ar xy$
and $\ar uv$ share an endpoint, we get a contradiction that we have reverted
the order on the second coordinate.

Now, we get the required drawing in the following way. For any arrow $\ar xy$ we
connect $x$ with the point $(x, y)$ and $y$ with $(y, x)$. We can do
all the connections simultaneously for all arrows without introducing any crossing since we
have respected the cyclic order on the first coordinate. We remove the interior
of $E_1$ and we identify the boundary. This way we introduce a crosscap.
Finally, we glue another disk along its boundary to $Z$ and we get the required
drawing on $\RP^2$.
\begin{figure}
\begin{center}
  \includegraphics{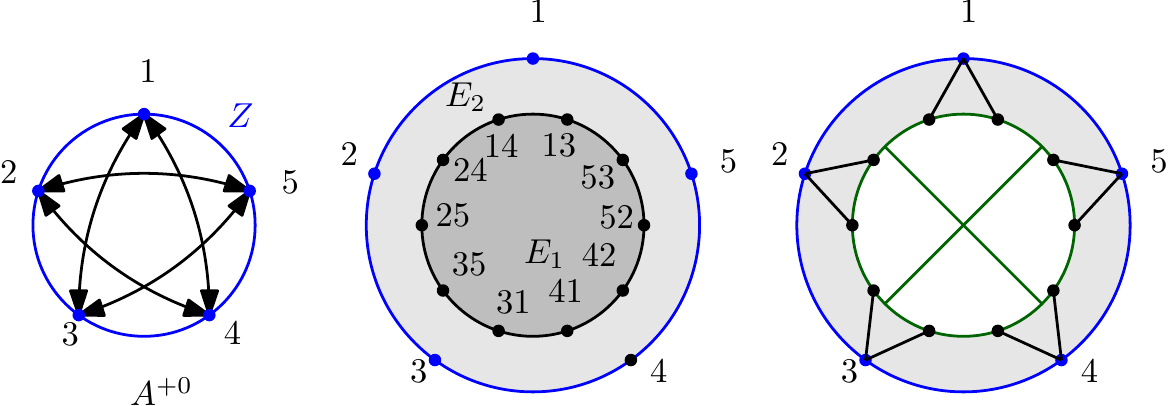}
  \caption{Redrawing the case where every inside bridge induces a single
  arrow.}
  \label{f:only_nontrivial}
\end{center}
\end{figure}
\end{proof}
\fi

\ifconf
\else

Finally, we prove Theorem~\ref{thm:main}.
\begin{proof}[Proof of Theorem~\ref{thm:main}]
 We prove the result by induction in the number of vertices of $G$. We can
 trivially assume that $G$ has at least three vertices.

 If $G$ has at least two blocks of 2-connectivity, $G$ can be written as $G_1\cup G_2$,
 where $G_1\cap G_2$ is a minimal cut of $G$ and, therefore, has at most one
 vertex. 
 By Lemma~\ref{l:2-connected} we may assume that $G_1$ is planar and $G_2$ non-planar.
 By induction, there exists an embedding $D_2$ of $G_2$ into $\R P^2$.
 So $G_1$ is planar, $G_2$ is embeddable into~$\RP^2$ and $G_1\cap G_2$ has
 at most one vertex.  From these two embeddings, we easily derive an
 embedding of  $G=G_1\cup G_2$ in~$\RP^2$.

We are left with the case when $G$ is $2$-connected. By Observation~\ref{obs:no_trivial_cycle},
we may assume that there is at least one trivial cycle $Z$ in $(D,\lambda)$.
We can also make each of its edges trivial by Lemma~\ref{l:planarize} 
and even by Lemma~\ref{lem:even}. Then we make $Z$, in addition, simple using Lemma~\ref{lem:simple}.
Hence $G$, $Z$ and the current projective HT-drawing satisfy the
separation assumptions.

Then we use $Z$ to
redraw $G$ as follows. At first, we apply Theorem~\ref{t:black_box} to get a projective HT-drawing $(D', \lambda')$ 
that separates $\gin$ and $\gout$. We define $\inn{D}:=D'(\gin)$ and $\out{D}:=D'(\gout)$---without loss of generality, 
$\out{D}$ is an ordinary HT-drawing on $S^2$, while $\inn{D}$ is a projective HT-drawing on $S^2$. 

Finally, we apply 
Proposition~\ref{p:inductive_hypothesis} above 
to $\inn{D}$ and $\out{D}$ separately. Thus, we get embeddings of $\gin$ and $\gout$---one of them in $S^2$, the
other one in $\RP^2$. In addition, $Z$ bounds a face in both of them; hence, 
we can easily glue them to get an embedding of the whole graph $G$ into $\RP^2$.
\end{proof}

\fi

\ifconf
\else 
\section{Labellings of Inside/Outside Bridges and the~Proof of
Pro\-position~\ref{p:redrawable_exist}}
\label{s:labellings}

In this section, given an inside (or outside) bridge $B$, we first describe
what are possible combinations of arrows induced by $B$. Then we use the
obtained findings for a proof of Proposition~\ref{p:redrawable_exist}, assuming
validity of Lemmas~\ref{lem:arrows_touch}, \ref{lem:disjoint_arrows_interleave}
and~\ref{l:no_triangle_arrows} which will be proved in Sect.~\ref{s:forbidden_arrows}.

\paragraph{Labelling the vertices of the inside/outside bridges.}
We start with the first step. As usual, we only describe
the `inside' case; the `outside' case will be analogous. We introduce certain
labellings of $V(B) \cap V(Z)$ which will help us to determine arrows.

\begin{definition}[Labelling of $V(B)\cap V(Z)$]\label{def:labels}
  A \emph{valid} labelling $L = L_B$ for $B$ is a mapping $L\colon V(B)\cap
  V(Z)\rightarrow\{\{0\},\{1\},\{0,1\}\}$ obtained in the following way.

If $V(B) \setminus V(Z) \neq \emptyset$ we pick a \emph{reference vertex} $v_B \in V(B) \setminus
V(Z)$ for $L$. Then we fix a \emph{labelling parameter} $\alpha_B \in \Z_2$ for $L$.
Finally, for any $u \in V(B)\cap V(Z)$ and for any proper walk $\omega$ with endpoints $u$
and $v_B$, the vertex $u$ receives the
label $\alpha_B + \lambda(\omega) \in \Z_2$. Note that $u$ may receive two labels after considering all such walks. On the other hand, each vertex of $V(B) \cap V(Z)$ obtains at least one label, which follows from the definition of bridges (Definition~\ref{d:bridge}).

If $V(B) \subseteq V(Z)$, then $B$ comprises only of one edge $e = uv$
  connecting two vertices of $V(Z)$. In such case, there are two valid
  labellings for $B$. We set $L(u) = \{\alpha_B\}$ and $L(v) = \{\lambda(e) +
  \alpha_B\}$ for a chosen labelling parameter $\alpha_B \in \Z_2$.
\end{definition}

If the bridge $B$ is understood from the context we may write just $v$ instead of $v_B$ for the reference vertex and $\alpha$ instead of $\alpha_B$ for the labelling parameter. By alternating the choice of $\alpha$ in the definition we may swap all labels.
This means that there are always at least two valid labellings for a given inside bridge.
On the other hand, a different choice of the reference vertex either does not influence the resulting labelling,
or has the same effect as swapping the value of the labelling parameter $\alpha$.
In other words, there are always exactly two valid labellings of the given inside/outside bridge $B$ corresponding to two possible choices of the labelling parameter $\alpha$, as is explained below.

To see this, consider a vertex $u\in V(B)\setminus V(Z)$ different from $v=v_B$. By Definition~\ref{d:bridge}, there is a proper $uv$-walk $\gamma$ in $B$ not using any vertex of $Z$. Now, for any $x\in V(B)\cap V(Z)$ and for any proper $xv$-walk $\omega_{xv}$ in $B$, the concatenation of the walks $\omega_{xv}$ and $\gamma$ is a proper $xu$-walk in $B$ of type $\lambda(\omega_{xv})+\lambda(\gamma)$. Also, for any proper $xu$-walk $\omega_{xu}$ in $B$, the concatenation of the walks $\omega_{xu}$ and $\gamma$ is a proper $xv$-walk in $B$ of type $\lambda(\omega_{xu})+\lambda(\gamma)$.
As a result, choosing $u$ as the reference vertex with $\alpha+\lambda(\gamma)$ as the labelling parameter leads to the same labelling as the choice of $v$ as the reference vertex with the labelling parameter $\alpha$.

The idea presented above can be used to establish the following simple observation, which we later use several times in the proofs.
\begin{observation}\label{o:reference_vertex}
Let $B$ be an inside or an outside bridge containing at least one inside/outside vertex. Moreover, let $L$ be a valid labelling for $B$ and $v$ the reference vertex for $L$. Let $x,y\in V(B)$ and let $\omega$ be a proper $xy$-walk in $B$. Then there is a proper $xy$-walk $\omega'$ in $B$ containing the reference vertex $v$ such that $\lambda(\omega)=\lambda(\omega')$.
\end{observation}
\begin{proof}
If $\omega$ contains inside/outside vertices, we choose one of them and denote it by $u$. If it does not contain any such vertex, then $x\in V(Z)$ and $x=y$, since $B$ cannot consist of just one edge. In this case we choose $u=x$.

Now we find a proper $uv$-walk $\gamma$ in $B$ and use it as a detour. More precisely, $\omega'$ starts at $x$ and follows $\omega$ to the first occurrence of $u$ in $\omega$. Then it goes to $v$ and back along $\gamma$. Finally, it continues to $y$ along $\omega$. It is clear that $\lambda(\omega)=\lambda(\omega')$. By the choice of $u$, the walk $\omega'$ is also proper.
\ifconf
\qed
\fi
\end{proof}

Now, whenever $u$ and $w$ are two vertices from $V(B) \cap V(Z)$, there is an
arrow $\ar{u}{w}$ arising from $B$ if and only if the vertices $u$ and $w$ were assigned
different labels by $L_B$---this is proved in Proposition~\ref{prop:arrows_labels} below.

\begin{proposition}\label{prop:arrows_labels}
Let $B$ be an inside bridge and $L$ be a valid labelling for $B$. Let $x, y
\in V(B) \cap V(Z)$ (possibly $x = y$).
Then the
inside arrow graph $A^+$ contains an arrow $\ar xy$ arising from $B$ if and only if $L(x) \cup L(y) = \{0,1\}$.
\end{proposition}

\begin{proof}
  It is straightforward to check the claim if $B$ is just an edge $e$. Indeed, if
  $x \neq y$, then $e=xy$, and it defines the arrow $\ar xy$ arising from $B$
  if and only if $\lambda(e)=1$, which in turn happens if and only if $L(x)\cup
  L(y)=\{0,1\}$ according to Definition~\ref{def:labels}. If $x = y$, then
  $\ar xx$ is not induced by $B$ and $|L(x) \cup L(x)| = 1$.

If $V(B) \setminus V(Z) \neq \emptyset$, let $v=v_B$ be the reference vertex for $L$.
First, let us assume that $L(x) \cup L(y) =
\{0,1\}$. Let us consider a proper $xv$-walk $\omega_{xv}$ and a proper $vy$-walk $\omega_{vy}$ in $B$ 
such that $\lambda(\omega_{xv}) \neq \lambda(\omega_{vy})$. Such walks
exist by Definition~\ref{def:labels}, since $L(x) \cup L(y) =
\{0,1\}$. Then the concatenation of these two walks is a nontrivial walk which
belongs to $W^+_{xy,B}$; therefore, $\ar xy$ is induced by $B$.

On the other hand, let us assume that there is a nontrivial walk $\omega$
in $W^+_{xy,B}$ defining the arrow $\ar xy$. We can assume that $\omega$ is not just an
edge, because it would mean that $B$ consists only of that edge.
By Observation~\ref{o:reference_vertex}, we may assume that $\omega$ contains the reference vertex $v$.
This vertex splits $\omega$ into two proper walks $\omega_1$ and $\omega_2$ so that each of them has at least
one edge. Since $\lambda(\omega)=1$, we have $\lambda(\omega_1)\neq \lambda(\omega_2)$.
Consequently, $L(x) \cup L(y) = \{0, 1\}$.
\ifconf
\qed
\fi
\end{proof}

The argument from the last two paragraphs of the proof above can also be used
to establish the following lemma.

\begin{lemma}
\label{l:single_labels}
Let $B$ be an inside or an outside bridge, let $L$ be a valid labelling for
$B$, and let $x,y \in V(B) \cap V(Z)$ be two distinct vertices. Moreover, we
assume that $|L(x)| = |L(y)| = 1$. Then for any proper $xy$-walks $\omega_1$,
$\omega_2$ in $B$ we have $\lambda(\omega_1) = \lambda(\omega_2)$.
\end{lemma}

\begin{proof}
If $B$ contains just the edge $xy$, the observation is trivially true.
Therefore, we assume that there is the inside/outside reference vertex $v\in V(B)$ for $L$.
By the assumption, every two proper $xv$-walks in $B$
have the same $\lambda$-value. The same holds also for proper $vy$-walks in $B$.
By Observation~\ref{o:reference_vertex}, we can assume that
both $\omega_1$ and $\omega_2$ contain $v$. Then the lemma follows.
\ifconf
\qed
\fi
\end{proof}

We will also need the following description of inside arrows induced by an inside bridge which
does not induce any loop.

\begin{lemma}
\label{l:bipartite}
 Let $B$ be an inside bridge which does not induce any loop. 
 Then the inside arrows induced by $B$ form a complete bipartite graph. (One of
 the parts is empty if $B$ does not induce any arrow.)
\end{lemma}

\begin{proof}
  Let us consider a valid labelling $L$ for $B$. By
  Proposition~\ref{prop:arrows_labels}, $\abs{L(x)} = 1$ for any $x \in V(B) \cap
  V(Z)$, since $B$ does not induce any loop. By
  Proposition~\ref{prop:arrows_labels} again, the inside arrows induced by $B$
  form a complete bipartite graph, in which one part corresponds to the vertices labelled
  $0$ and the second part corresponds to the vertices labelled $1$.
\ifconf
\qed
\fi
\end{proof}

We conclude this section a by a proof of Proposition~\ref{p:redrawable_exist}.

\begin{proof}[Proof of Proposition~\ref{p:redrawable_exist}]
We need to distinguish few cases. 

First, we consider the case when we have two disjoint inside arrows, but at least one of
them is a loop. In this case, it is easy to see that
Lemma~\ref{lem:arrows_touch} implies that $G$ forms the
outside fan and we are done.

\bigskip

Second, let us consider the case that we have two disjoint inside arrows $\ar
ab$ and $\ar cd$ which are not loops. Lemma~\ref{lem:arrows_touch} implies that
the only possible outside arrows are $\ar ac$, $\ar ad$, $\ar bc$, $\ar bd$.
(In particular, there are no loops outside.) If there are not two disjoint arrows
outside, then $G$ forms an outside fan and we are done.
Therefore, we may assume that there are two disjoint arrows outside, without
loss of generality, $\ar ac$ and $\ar bd$ (otherwise we swap $a$ and $b$).
By swapping outside and inside in the previous argument, we get that only
further possible arrows inside are $\ar ad$ and $\ar bc$.

Now we distinguish a subcase when there is an inside bridge inducing the
inside arrows $\ar ab$ and $\ar cd$. In this case, $\ar ad$ and $\ar bc$ must be
inside arrows as well by Lemma~\ref{l:bipartite}.
By Lemma~\ref{lem:arrows_touch}, we know that $\ar ac$ and $\ar bd$ are the only
outside arrows and we get that they must alternate by
Lemma~\ref{lem:disjoint_arrows_interleave}. That is, up to relabelling of the
vertices, we get the right cyclic order for an inside square. In order to check
that $G$ indeed forms an inside square, it remains to verify that $G$ has only
one nontrivial inside bridge. The inside arrows are $\ar ab$, $\ar bc$, $\ar
cd$ and $\ar ad$. If any of these arrows, for example $\ar ab$, is induced by
two bridges, then we get a contradiction with
Lemma~\ref{lem:disjoint_arrows_interleave}, in this case on arrows $\ar ab$ and
$\ar cd$.

By swapping inside and outside we solve the subcase when there is an outside
bridge inducing the outside arrows $\ar ac$ and $\ar bd$; we get that $G$ forms
an outside square.

It remains to consider the subcase when $\ar ab$ and $\ar cd$ arise from
different inside bridges and $\ar ac$ and $\ar bd$ arise from different
outside bridges. However, Lemma~\ref{lem:disjoint_arrows_interleave} applied
to the inside and then to the outside reveals that these two events cannot
happen simultaneously.

Consequently, we have proved Proposition~\ref{p:redrawable_exist} in case there are two disjoint inside arrows. Analogously, we resolve the case when we have two disjoint arrows outside.

\bigskip

Finally, we consider the case when every pair of inside arrows
shares a vertex and every pair of outside arrows shares a vertex. If there is a
vertex $v$ common to all the inside arrows, then we get an inside fan and we are done.

It remains to consider the last subcase when there is no vertex common to all
inside arrows while every pair of inside arrows shares a vertex. This leaves
the only option that there are three distinct vertices $a$, $b$ and $c$ on $Z$ and all
three inside arrows $\ar ab, \ar ac$ and $\ar bc$ are present. Then, the
only possible outside arrows are $\ar ab, \ar ac$ and $\ar bc$ as well due
to Lemma~\ref{lem:arrows_touch}. In addition, all three outside arrows $\ar
ab$, $\ar ac$ and $\ar bc$ must be present, otherwise we have an outside fan 
and we are done.

In the present case, an inside bridge can induce at most two
arrows by Lemma~\ref{l:bipartite}. Let us consider the three pairs
of arrows $\{\ar ab, \ar ac\}$, $\{\ar ab, \ar bc\}$, and $\{\ar ac, \ar bc\}$.
If at most one of these pairs is induced by an inside bridge, then $G$ forms
an inside split triangle and we are done. Analogously, we are done, if at most
one of these pairs is induced by an outside bridge. Therefore, it remains to
consider the case that at least two such pairs are induced by inside bridges
and at least two such pairs are induced by outside bridges. However, this
yields a contradiction to Lemma~\ref{l:no_triangle_arrows}.
\ifconf
\qed
\fi
\end{proof}

\section{Forbidden Configurations of Arrows}
\label{s:forbidden_arrows}

In this section we show that certain combinations of arrows are not possible.
That is, we prove Lemmas~\ref{lem:arrows_touch},
\ref{lem:disjoint_arrows_interleave} and~\ref{l:no_triangle_arrows}.
As before, we have a fixed graph $G$, its drawing $(D,\lambda)$ on $S^2$ and a cycle $Z$ in $G$.
Again, we assume that $G, (D,\lambda)$ and $Z$ satisfy the separation assumptions.

\paragraph{Homology and intersection forms.}
We start with a brief explanation of intersection forms that will help us to
prove the required lemmas.

We assume that the reader is familiar with basics of homology theory, otherwise
we refer to the introductory books by Hatcher~\cite{Hatcher2002} or Munkres~\cite{Munkres2000}. 
We always work with homology
over $\Z_2$ and, unless stated otherwise, we work with singular homology.
Let $S$ be a surface. We will mainly work with the first homology group
and we denote by $B_1(S)$, $Z_1(S)$ and $H_1(S) := Z_1(S)/B_1(S)$ the
group of \emph{$1$-boundaries}, of \emph{$1$-cycles} and the first
\emph{homology group}, respectively. Given a $1$-cycle $z \in Z_1(S)$, if there
is no risk of confusion, we also consider it as an element of $H_1(S)$,
although, formally speaking, we should consider its homology class $[z]$.
Similarly, if there is no risk of confusion, we do not distinguish a 1-cycle
and its support. Namely, by an intersection of two $1$-cycles we actually mean
an intersection of their images. We use the same convention for crossings, that
is, transversal intersections.

Let $S$ be a surface. The \emph{intersection form} on $S$ is a unique bilinear
map $\Omega_S\colon H_1(S) \times H_1(S) \to \Z_2$ with the following property.
Whenever $z_1, z_2 \in Z_1(S)$ are two $1$-cycles intersecting in finite number
of points and crossing in every such point (i.~e., intersecting transversally), then $\Omega_S(z_1, z_2)$ is the
number of crossings of $z_1$ and $z_2$ modulo $2$; we refer to \cite[Sect. 8.4]{Fuchs2004} for the existence of $\Omega_S$.
In particular,
$\Omega_{S^2}$ is the trivial map since $H_1(S^2)$ is trivial. On the other
hand, $\Omega_{\RP^2}$ is already nontrivial:

\begin{lemma}[Intersection form on $\R P^2$]\label{lem:2cycles}
Let $z_1$ and $z_2$ be two homologically nontrivial $1$-cycles in $\R P^2$.
Then $\Omega_{\RP^2}(z_1, z_2) = 1$. In particular, if $z_1$ and $z_2$ have a
finite number of intersections and they cross at every intersection, then they
have to cross an odd number of times.
\end{lemma}
\begin{proof}
  Since the intersection form $\Omega_{\RP^2}$ depends only on the homology
  class, and since $H_1(\RP^2) = \Z_2$, it is sufficient to exhibit any two nontrivial $1$-cycles that
  intersect an odd number of times on $\RP^2$. This is an easy task.
\ifconf
\qed
\fi
\end{proof}

\paragraph{From sphere to the projective plane.} Although it is overall simpler
to do the proof of Theorem~\ref{thm:main} in the setting of
projective HT-drawings on $S^2$, it is easier to prove
Lemmas~\ref{lem:arrows_touch},
\ref{lem:disjoint_arrows_interleave} and~\ref{l:no_triangle_arrows} in the
setting of HT-drawings on $\RP^2$. A small drawback is that we need to check
that splitting of $S^2$ to the inside and outside part works analogously on
$\RP^2$ as well.

\begin{lemma}
  \label{l:split_RP2}
  Let $(D, \lambda)$ be a projective HT-drawing of a graph $G$ on $S^2$ and let 
  $Z$ be a cycle satisfying the separation assumptions.
  Let $D_\otimes$ be the
  HT-drawing of $G$ on $\RP^2$ coming from the proof of
  Lemma~\ref{l:S2_to_RP2}. Then $D_\otimes(Z)$ is a simple cycle such that
  each of its edges is even, which
  separates $\RP^2$ into two parts, $(\RP^2)^+$ and $(\RP^2)^-$. In addition,
  every inside edge (with respect to $D$) which is incident to a vertex of $Z$
  points locally into $(\RP^2)^+$ in $D_\otimes$ as well as every outside edge
  (with respect to $D$) which is incident to a vertex of $Z$ points locally
  into $(\RP^2)^-$.
\end{lemma}

\begin{proof}
  By statement of Lemma~\ref{l:S2_to_RP2} we already know that $D_\otimes(Z)$
  is a simple cycle and that each of its edges is even. For the rest, we need to inspect the construction
  of $D_\otimes$ in the proof of Lemma~\ref{l:S2_to_RP2}. However, we get all
  the required conclusions directly from this construction.
\ifconf
\qed
\fi
\end{proof}

\paragraph{Drawings of walks.}
We also need to set up a convention regrading drawings of walks in a graph $G$.
Let $D$ be a drawing of a graph $G$ on a surface $S$. Let $\omega$ be a walk in
$G$. Then $D$ induces a continuous
map $D(\omega)\colon [0,1] \to S$; it is given by the concatenation of drawings
of edges of $\omega$. Here we also allow that $\omega$ is a walk of length $0$
consisting of a single vertex $v$. Then $D(\omega)$ is a constant map whose
image is $D(v)$. If $\omega$ is a closed walk, then we may regard it as an
element of $H_1(S)$. 

\paragraph{Proofs of the lemmas.} 
Now we have introduced enough tools to prove the required lemmas. In all three
proofs, $D_\otimes$ stands for the HT-drawing on $\RP^2$ from Lemma~\ref{l:split_RP2}.
First, we prove
Lemma~\ref{lem:disjoint_arrows_interleave} which has a very simple proof.
In fact, we prove slightly stronger statement which we plan to reuse later on.
Lemma~\ref{lem:disjoint_arrows_interleave} follows directly from
Lemma~\ref{l:noninterleaving_walks_intersect} below.

\begin{lemma}
\label{l:noninterleaving_walks_intersect}
Let $a$, $b$, $x$ and $y$ be four distinct vertices of $Z$ such that $x$ and $y$ are on
the same arc of $Z$ when split by $a$ and $b$. Then any two walks
$\omega^+_{ab} \in W^+_{ab}$ and $\omega^+_{xy} \in W^+_{xy}$ must share a
vertex.
\end{lemma}

\begin{proof}

We consider a closed walk $\kappa^+_{ab}$ arising
from a concatenation of the walk $\omega^+_{ab}$ and the arc of $Z$ connecting $a$ and $b$ not containing
$x,y$. We also consider the closed walk $\kappa^+_{xy}$ obtained
analogously. 
See Fig.~\ref{f:non_interleaving}. 
The homological 1-cycles corresponding to
$D_\otimes(\kappa^+_{ab})$ and $D_\otimes(\kappa^+_{xy})$ are both non-trivial;
therefore, by Lemma~\ref{lem:2cycles}, $D_\otimes(\kappa^+_{ab})$ and
$D_\otimes(\kappa^+_{xy})$ must have an odd number of crossings.
(Note
that, for example, $D_\otimes(\kappa^+_{ab})$ may have self-intersections or
self-touchings, but there is a finite number of intersections between
$D_\otimes(\kappa^+_{ab})$ and
$D_\otimes(\kappa^+_{xy})$ which are necessarily crossings.)
However, if $\omega^+_{ab} \in W^+_{ab}$ and $\omega^+_{xy} \in W^+_{xy}$ did
not have a vertex in common, then $D_\otimes(\kappa^+_{ab})$ and $D_\otimes(\kappa^+_{xy})$
would have an even number of crossings, because $D_\otimes$ is an HT-drawing by
Lemma~\ref{l:S2_to_RP2}.
\ifconf
\qed
\fi
\end{proof}

\begin{figure}
\begin{center}
  \includegraphics{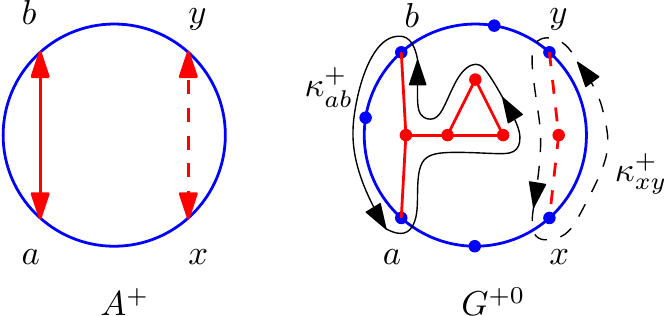}
\caption{Walks in Lemma~\ref{l:noninterleaving_walks_intersect}.}
\label{f:non_interleaving}
\end{center}
\end{figure}

We have proved Lemma~\ref{lem:disjoint_arrows_interleave} and we continue with
the proofs of the next two lemmas.

\begin{proof}[Proof of Lemma~\ref{lem:arrows_touch}]
To the contrary, we assume that we have an inside arrow $\ar{x}{y}$ and an
outside arrow $\ar{u}{v}$ which do not share any endpoint. However, we allow $x
= y$ or $u = v$, that is, we allow loops.
As before, we consider a closed walk $\kappa^+_{xy}$ obtained from the
concatenation of a walk from $\omega^+_{xy} \in W^+_{xy}$ and any of the two arcs of $Z$ connecting $x$
and $y$. If $x = y$, then we do not add the arc from $Z$.
Analogously, we have a closed walk $\kappa^-_{uv}$ coming from a walk in
$W^-_{uv}$ and an arc of $Z$ connecting $u$ and $v$. Both of these
walks are nontrivial and we aim to get a contradiction with
Lemma~\ref{lem:2cycles}.

Unlike the previous proof, this time $D_\otimes(\kappa^+_{xy})$ and
$D_\otimes(\kappa^-_{uv})$ may not cross at every intersection.
Namely, $\kappa^+_{xy}$ and $\kappa^-_{uv}$ may share some subpath of $Z$, but
apart from this subpath the intersections are crossings. We slightly modify
these drawings in the following way. Let us recall that
$D_\otimes(Z)$ splits $\RP^2$ into two parts $(\RP^2)^+$ and $(\RP^2)^-$ according to
Lemma~\ref{l:split_RP2}. We slightly push into $(\RP^2)^+$ the subpath of
$\kappa^+_{xy}$ shared with $Z$ (possibly consisting of a single vertex). 
This way, we obtain a drawing $D_\otimes^+$ of $\kappa^+_{xy}$.
Similarly, we slightly push the subpath of $\kappa^-_{uv}$ shared with $Z$ into
$(\RP^2)^-$, obtaining a drawing $D_\otimes^-$ of $\kappa^-_{uv}$. See
Fig.~\ref{f:non_disjoint}.
Now,
$D_\otimes^+(\kappa^+_{xy})$ and $D_\otimes^-(\kappa^-_{uv})$ cross at
every intersection and the crossings of $D_\otimes^+(\kappa^+_{xy})$ and
$D_\otimes^-(\kappa^-_{uv})$ correspond to the crossings of
$D_\otimes(\kappa^+_{xy})$ and
$D_\otimes(\kappa^-_{uv})$.

We now consider the crossings of $D_\otimes(\kappa^+_{xy})$ and
$D_\otimes(\kappa^-_{uv})$. Whenever $e$ is an edge of $\kappa^+_{xy}$ and $f$
is an edge of $\kappa^-_{uv}$ such that $e$ and $f$ are independent, then
$D_\otimes(e)$ and $D_\otimes(f)$ have an even number of crossings, because
$D_\otimes$ is an HT-drawing. However, if $e$ and $f$ are adjacent, then they
still cross evenly since one of these edges must belong to $Z$. Here we
crucially use that $\ar xy$ and $\ar uv$ do not share any endpoint. Therefore,
$D_\otimes(\kappa^+_{xy})$ and $D_\otimes(\kappa^-_{uv})$ have an even number
of crossings, and consequently, $D_\otimes^+(\kappa^+_{xy})$ and
$D_\otimes^-(\kappa^-_{uv})$ as well. This is a contradiction to
Lemma~\ref{lem:2cycles}.
\begin{figure}
\begin{center}
  \includegraphics{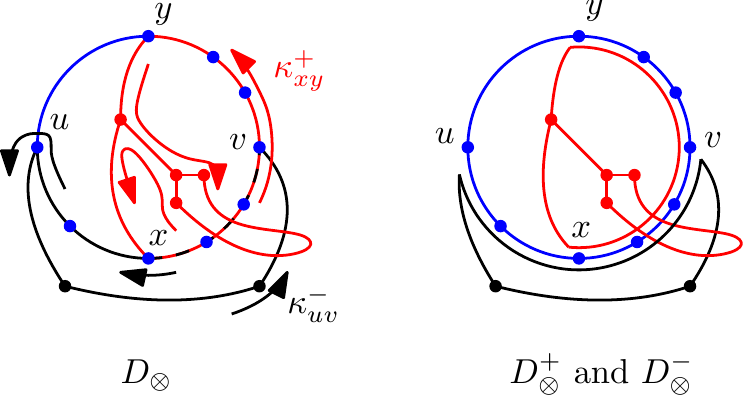}
\caption{Walks in Lemma~\ref{lem:arrows_touch}.}
\label{f:non_disjoint}
\end{center}
\end{figure}
\ifconf
\qed
\fi
\end{proof}

\begin{proof}[Proof of Lemma~\ref{l:no_triangle_arrows}]
For contradiction, there is such a configuration.
  
Let $e_a^+$ be any edge of $E(B^+)$ incident to $a$. Analogously, we define
edges $e_a^-$, $e_b^+$, $e_b^-$, $e_c^+$ and $e_c^-$. 
We observe that there is a walk $\omega^+_{ab} \in W^+_{ab}$ which uses the edges $e_a^+$ and $e_b^+$.
Indeed, it is sufficient to consider arbitrary proper walk using $e_a^+$ and
$e_b^+$ in $B^+$. This walk is nontrivial by Lemma~\ref{l:single_labels}. (The
assumptions of the lemma are satisfied by Proposition~\ref{prop:arrows_labels}
since $B^+$ does not induce any inside loops.)
We also let $\kappa^+_{ab}$ be the closed walk obtained from the
concatenation of $\omega^+_{ab}$ and the arc of $Z$ connecting $a$ and $b$ and
avoiding $c$. Analogously, we define $\omega_{ac}^+$, $\omega_{ab}^-$,
$\omega_{bc}^-$ and closed walks $\kappa_{ac}^+$, $\kappa_{ab}^-$ and
$\kappa_{bc}^-$. When defining the closed walks, we always use the arc of $Z$
which avoids the third point among $a$, $b$ and $c$. All these eight walks are nontrivial.

Now, we aim to show that $e_a^+$ and $e_a^-$ cross oddly in the drawing
$D_\otimes$. We consider the closed walks $\kappa_{ab}^-$ and $\kappa_{ac}^+$
and their drawings $D_\otimes(\kappa_{ab}^-)$ and $D_\otimes(\kappa_{ac}^+)$.
The walks $\kappa_{ab}^-$ and $\kappa_{ac}^+$ share only the point $a$;
therefore, $D_\otimes(\kappa_{ab}^-)$ and $D_\otimes(\kappa_{ac}^+)$ cross at
every intersection possibly except $D_\otimes(a)$. By Lemma~\ref{l:split_RP2}
we know that $e_a^+$ and $e_a^-$ point to different sides of $Z$ (in
$D_\otimes$); thus, $D_\otimes(\kappa_{ab}^-)$ and
$D_\otimes(\kappa_{ac}^+)$ actually touch in $D_\otimes(a)$. This touching can
be removed by a slight perturbation of these cycles, analogously as in the proof
of Lemma~\ref{lem:arrows_touch}, without affecting other intersections. By
Lemma~\ref{lem:2cycles} we therefore get that $D_\otimes(\kappa_{ab}^-)$ and
$D_\otimes(\kappa_{ac}^+)$ have an odd number of crossings. However, if we
consider any pair of edges $(e, f)$ where $e$ is an edge of $\kappa_{ab}^-$ and
$f$ is an edge of $\kappa_{ac}^+$ different from $(e_a^-, e_a^+)$, we get that
$e$ and $f$ cross an even number of times. Indeed, if we have such 
$(e,f) \neq (e_a^-, e_a^+)$, then either $e$ or $f$ belongs to $Z$, or they are
independent. Consequently, the odd number of crossings of
$D_\otimes(\kappa_{ab}^-)$ and
$D_\otimes(\kappa_{ac}^+)$ has to be realized on $e_a^+$ and $e_a^-$.

Analogously, we show that $e_b^+$ and $e_b^-$ must cross oddly by considering
the walks $\kappa_{ab}^+$ and $\kappa_{bc}^-$.

Now let us consider the closed walk $\kappa_{ab}^+$ and a closed walk $\mu_{ab}^-$
obtained from the concatenation of $\omega_{ab}^-$ and the arc of $Z$
connecting $a$ and $b$ which contains $c$. By analogous ideas as before, we get
that $D_\otimes(\kappa_{ab}^+)$ and $D_\otimes(\mu_{ab}^-)$ touch in
$D_\otimes(a)$ and $D_\otimes(b)$; if they intersect anywhere else, they cross there.
Using a small perturbation as before, they must have an odd number of crossings by Lemma~\ref{lem:2cycles}.
On the other hand, the pairs of edges $(e_a^+, e_a^-)$ and
$(e_b^+, e_b^-)$ cross oddly, as we have already observed. Any other pair $(e,
f)$ of edges where $e$ is an edge of $\kappa_{ab}^+$ and $f$ is an edge of
$\mu_{ab}^-$ must cross evenly since they are either independent or one of them
belongs to $Z$. This means that $D_\otimes(\kappa_{ab}^+)$ and
$D_\otimes(\mu_{ab}^-)$ intersect evenly, which is a contradiction.
\ifconf
\qed
\fi
\end{proof}

\paragraph{Intersection of trivial interleaving walks.}
We conclude this section by a proof of a lemma similar in spirit to
Lemma~\ref{l:noninterleaving_walks_intersect}. We will need this Lemma in
Sect.~\ref{s:redrawings}, but we keep the lemma here due to its similarity to
previous statements.

\begin{lemma}
\label{l:interleaving_walks_intersect}
Let $a$, $b$, $x$ and $y$ be four distinct vertices of $Z$ such that $x$ and $y$ are on
different arcs of $Z$ when split by $a$ and $b$. Let $\omega^+_{ab}$ and $\omega^+_{xy}$
be a proper $ab$-walk and a proper $xy$-walk in $\gin$, respectively, such that $\lambda(\omega^+_{ab}) =
\lambda(\omega^+_{xy}) = 0$. 
Then $\omega^+_{ab}$ and $\omega^+_{xy}$ must share a
vertex.
\end{lemma}

\begin{proof}
We proceed by contradiction. As usual, we consider closed walks $\kappa_{ab}^+$ and $\kappa_{xy}^+$ defined as follows.
The walks $\kappa_{ab}^+$ consists of $\omega_{ab}^+$ and an arc of $Z$ connecting $a$ and $b$, while the walk $\kappa_{xy}^+$ is formed by $\omega_{xy}^+$ and an arc of $Z$ connecting $x$ and $y$. This time, $\omega_{ab}^+$ and $\omega_{xy}^+$ are trivial.

 We push $D_\otimes(\kappa_{ab}^+)$ a bit inside and $D_\otimes(\kappa_{xy}^+)$
 a bit outside of $Z$, similarly as in the proof of
 Lemma~\ref{lem:arrows_touch}. This time, however, we introduce one more
 crossing, because both $\kappa_{ab}^+$ and $\kappa_{xy}^+$ are walks in $\gin$.
 Since the intersection form of trivial cycles corresponding to the drawings of
 $\kappa_{ab}^+$ and $\kappa_{xy}^+$ is trivial, we get that these drawings
 have to cross an even number of times.
 This in turn means that the drawings of $\omega_{ab}^+$ and $\omega_{xy}^+$ cross an odd number of times.
 Since $D_\otimes$ is an HT-drawing,
 this yields a contradiction to the assumption that 
 $\omega_{ab}^+$ and $\omega_{xy}^+$ do not share a vertex.
\ifconf
\qed
\fi
\end{proof}

\section{Redrawings}
\label{s:redrawings}

We will prove Proposition~\ref{p:redrawings} in this section separately for each case.
That is, we show that if $\gin$ forms any of the configurations depicted in Fig.~\ref{f:redrawable},
then $\gin$ admits an ordinary HT-drawing on $S^2$.
However, we start with a general redrawing result that we will use in all cases.

\begin{lemma}\label{l:labels}
  Let $(D, \lambda)$ be a projective HT-drawing of $\gin$ on $S^2$ and $Z$ a
  cycle satisfying the separation assumptions. 
Let us also assume that that $D(\gin) \cap S^-
= \emptyset$.
Let $B$ be one of the inside bridges different from an edge and let $L$ be a valid labelling of
$B$. Let us assume that there is at least one vertex $x \in V(B) \cap V(Z)$
such that $|L(x)| = 1$. 
Then there is a projective HT-drawing $(D', \lambda')$ of $\gin$ on $S^2$ such that 

\begin{enumerate}[$(a)$]
  \item $D$ coincides with $D'$ on $Z$ and $D'(\gin) \cap S^- = \emptyset$;
  \item every edge $e \in E(\gin) \setminus E(B)$ satisfies $\lambda(e) =
    \lambda'(e)$;
  \item every edge $e\in E(B)$ that is not incident to $Z$ satisfies
    $\lambda'(e)=0$; and
  \item for every edge $uv=e\in E(B)$ such that $u\in V(Z)$, 
    we have $\lambda'(e)\in L(u)$.
\end{enumerate}
\end{lemma}

Note that the condition $(b)$ allows that the edges in inside
bridges other than $B$ may be redrawn, but only under the condition, that their triviality/nontriviality
is not affected.

\begin{proof}
  Let $B^+$ be the subgraph of $B$ induced by the vertices of $V(B) \setminus
  V(Z)$. By the definition of the inside bridge, the graph $B^+$ is
  connected; it is also nonempty since we assume that $B$ is not an edge.

  Every cycle of the graph $B^+$ must be trivial. Indeed, if $B^+$ contained a
  nontrivial cycle, than this cycle could be used to obtain a nontrivial proper walk 
  from $x$ to $x$. This would contradict the fact that $|L(x)| = 1$ via
  Proposition~\ref{prop:arrows_labels}. That is, $B^+$ satisfies the
  assumptions of Lemma~\ref{l:planarize}. Let $U \subseteq V(B^+)$ be the set
  of vertices obtained from Lemma~\ref{l:planarize}. That is, if we perform the
  vertex-crosscap switches on $U$, we obtain a projective HT-drawing $(D_U,
  \lambda_U)$
  such that $\lambda_U(e) = 0$ for any edge $e \in E(B^+)$.

  Let us recall that every vertex-crosscap switch over a vertex $y$ 
  is obtained from vertex-edge switches of nontrivial edges over $y$ and then 
  from swapping the value of $\lambda$ on all edges incident to $y$. The vertex-edge switches do
  not affect the value of $\lambda$. Overall,
  we get that $D_U$ coincides with $D$ on $Z$. We also require that all vertex-edge
 switches are performed in $S^+$; therefore, $D_U$ does not reach $S^-$.
 Altogether,
 $D_U$ and $\lambda_U$ satisfy $(a)$, $(b)$ and $(c)$, but we do not know yet whether
 $(d)$ is satisfied.

 In fact, $(d)$ may not be satisfied and we still may need to modify $D_U$ and
 $\lambda_U$. Let $e_0$ be any edge incident with $x$. If $L(x) =
 \{\lambda_U(e_0)\}$, we set $D' := D_U$ and $\lambda' := \lambda_U$. If $L(x) \neq
  \{\lambda_U(e_0)\}$, we further perform vertex-crosscap switches over all
  vertices in $V(B^+)$, obtaining $D'$ and $\lambda'$. We want to check $(a)$
  to $(d)$ for $D'$ and $\lambda'$. 

  It is sufficient to check $(a)$, $(b)$ and $(c)$ only in the latter case. Regarding
 $(a)$, we again change the drawing only by vertex-edge switches over edges $e$
 with $\lambda_U(e) = 1$ inside $S^+$. Validity of $(b)$ is obvious from the fact
 that $\lambda_U$ may be changed only on edges incident with $V(B^+)$.
 Regarding $(c)$, for any edge $e \in
 E(B^+)$ we perform the vertex-crosscap switch for both endpoints of $e$.
 Therefore, $\lambda'(e) = \lambda_U(e) = 0$. It remains to check $(d)$.

 First, we realize that we have set up $D'$ and $\lambda'$ in such a way that
 $L(x) = \{\lambda'(e_0)\}$. Indeed, if $L(x) \neq
     \{\lambda_U(e_0)\}$, then we have made a vertex-crosscap switch over
     exactly one endpoint of $e_0$. In particular, we have just checked $(d)$
     if $e = e_0$.

     Now, let $e = uv \neq e_0$ be an edge from
     $(d)$. We need to check that $\lambda'(e) \subseteq L(u)$. If $L(u) = \{0,1\}$,
     then we are done; therefore, we may assume that $|L(u)| = 1$.
     Let
     $\omega$ be any proper $xu$-walk in $B$ containing $e_0$ and $e$. Such a walk
     exists from the definition of an inside bridge (see Definition~\ref{d:bridge}). 
     We have $\lambda(\omega) = \lambda'(\omega)$ because the vertex-crosscap
     switches over the inner vertices of $\omega$ do not affect the triviality of $\omega$. But
     we also have $\lambda'(\omega) = \lambda'(e_0) + \lambda'(e)$ because
     $\lambda'(f) = 0$ for any edge $f \in E(B^+)$. 
     Since $L(x) =
     \{\lambda'(e_0)\}$ and $|L(u)| = 1$, it follows that $L(u) =
     \{\lambda'(e)\}$ by Proposition~\ref{prop:arrows_labels} and
     Lemma~\ref{l:single_labels} applied to $x$ and $u$. 
\ifconf
\qed
\fi
\end{proof}

\paragraph{Inside fan.} Now we may prove Proposition~\ref{p:redrawings} for
inside fans, which is the simplest case.

\begin{proof}[Proof of Proposition~\ref{p:redrawings} for inside fans.]
  We assume that $\gin$ forms an inside fan; see
  Fig.~\ref{f:redrawable}. Let $x
  \in V(Z)$ be the endpoint common to all inside arrows.
  Let us consider any inside
  bridge $B$, possibly trivial. 
  Let $L = L_B$ be a valid labelling of $B$. It follows from
  Proposition~\ref{prop:arrows_labels} that $|L(u)| = 1$ for any $u \in V(B)
  \cap V(Z)$ different from $x$. (Actually, there is at least one such $u$,
  because we assume that $G$ is $2$-connected; this is contained in the
  separation assumptions.) In addition, all $u \in V(B)
    \cap V(Z)$ different from $x$ have to have the same labels, 
    because there are no arrows
  among them. Since we may switch all labels in a valid labelling by changing the value of the labelling parameter, we may
  assume that $L(u) = \{0\}$ for any such $u$.

  Now, we consider all inside bridges $B_1, \dots, B_\ell$ (possibly
  trivial) and
  the corresponding labellings $L_{B_1}, \dots L_{B_\ell}$ as above. We apply
  Lemma~\ref{l:labels} to each of these bridges which is not an edge one by one. 
  This way we get a
  projective HT-drawing $(D_1,\lambda_1)$ which satisfies:
  \begin{enumerate}[$(i)$]
\item $D$ coincides with $D_1$ on $Z$ and $D_1(\gin) \cap S^- = \emptyset$;
\item every edge $e \in E(\gin)$ which is not incident with $Z$ satisfies
  $\lambda_1(e) = 0$;
\item every edge $e\in E(\gin)$ such that $\lambda_1(e) = 1$ is incident with $x$.
\end{enumerate}

Indeed, property~$(i)$ follows from the iterative application of property $(a)$ of
Lemma~\ref{l:labels}. Property~$(ii)$ follows from the iterative application of
properties $(b)$ and $(c)$ of Lemma~\ref{l:labels}. Finally, property~$(iii)$
follows from $(ii)$, from the iterative application of properties $(b)$ and $(d)$ of
Lemma~\ref{l:labels} and from the fact that any nontrivial inside bridge
which is a single edge must contain $x$.

Finally, we set $D' := D_1$ and let $\lambda'\colon E(\gin) \to \{0,1\}$ be the
constantly zero function. We observe that from $(ii)$ and $(iii)$, it follows
that $\lambda'(e)\lambda'(f) = \lambda_1(e)\lambda_1(f)$ for any pair of
independent edges of $\gin$. Therefore $(D',\lambda')$ is a projective HT-drawing as well. But, since $\lambda'$ is
identically zero function, $D'$ is also just an ordinary HT-drawing on $S^2$.
\ifconf
\qed
\fi
\end{proof}

\paragraph{Inside square.} Now we prove Proposition~\ref{p:redrawings} for an
inside square. Let $B$ be the inside bridge inducing the inside square and
let $a$, $b$, $c$ and $d$ be the vertices of $V(B) \cap V(Z)$ labelled according
to Definition~\ref{d:redrawable}. The main ingredient for our proof of
Proposition~\ref{p:redrawings} is the following lemma, which shows that $B$ must
have a suitable cut vertex.

\begin{lemma}
\label{l:cut_vertex}
The inside bridge $B$, inducing the inside square, contains a vertex $v$
such that the graph $B - v$ is disconnected and the vertices $a$, $b$, $c$ and
$d$ belong to four different components of $B - v$.
\end{lemma}

We first show how Proposition~\ref{p:redrawings} for inside squares
follows from Lemma~\ref{l:cut_vertex}. The proof is analogous to the previous
proof.

\begin{proof}[Proof of Proposition~\ref{p:redrawings} for inside squares.]
We assume that $B$ is the unique inside bridge inducing the inside square
and $a$, $b$, $c$ and $d$ are vertices of $V(B) \cap V(Z)$ as above. In
addition, let $v$ be the vertex from Lemma~\ref{l:cut_vertex}.

First we consider valid labellings of trivial inside bridges. After possibly
switching the value of the labelling parameter, we may achieve that all labels of a trivial inside bridge
are $0$.
We apply Lemma~\ref{l:labels} to all trivial inside bridges (which are
not an edge) and we get a projective $HT$-drawing $(D_1, \lambda_1)$ such that
$\lambda_1(e) = 0$ for any edge of $\gin$ which does not belong to the
nontrivial bridge $B$. Also, we did not affect $\lambda$ on edges of $B$,
$D_1$ coincides with $D$ on $Z$ and we still have $D_1(\gin)\cap\out{S}=\emptyset$.

Now, we consider a valid labelling $L$ of $B$. It is easy to check that, up to
switching all labels, we have $L(a) = L(c) = \{1\}$ and $L(b) = L(d) = \{0\}$.
We apply Lemma~\ref{l:labels} to $B$ according to this labelling and we get a
projective HT-drawing $(D_2, \lambda_2)$ such that the only edges $e$ of $\gin$ with
$\lambda_2(e) = 1$ are the edges of $B$ incident to $a$ or $c$.

Next, let $C_a$ and $C_c$ be the components of $B - v$ which contains
$a$ and $c$, respectively. We perform vertex-crosscap switches over all vertices of $C_a$
and $C_c$ except $a$, $c$ and $v$. We perform the switches inside $S^+$ as
usual. This way we get a projective HT-drawing $(D_3, \lambda_3)$ such that only
edges $e$ of $\gin$ such that $\lambda_3(e) = 1$ are the edges of $B$ incident
to $v$.

Finally, we let $D' = D_3$ and we set $\lambda'(e) = 0$ for any edge $e$ of
$\gin$. Analogously as in the previous proof, $\lambda_3(e)\lambda_3(f) =
\lambda'(e)\lambda'(f)$ for any pair of independent edges of $\gin$. Therefore, 
$(D', \lambda')$ is a projective HT-drawing on $S^2$ and $D'$ is also an
ordinary HT-drawing on $S^2$, as required.
\ifconf
\qed
\fi
\end{proof}

It remains to prove Lemma~\ref{l:cut_vertex} to conclude the case of inside
squares.

We start with a certain separation lemma in a general graph 
and then we conclude the proof by
verification that the assumptions of this lemma are satisfied.

\begin{lemma}
\label{l:intersecting_paths}
  Let $G'$ be an arbitrary connected graph and $A = \{a_1, \dots, a_4\} \subseteq V(G')$ be
  a set of four distinct vertices. Let us assume that any $a_ia_j$-path has a
  common point in $V(G') \setminus A$ with any $a_ka_\ell$-path whenever
  $\{i,j,k,\ell\} = \{1,2,3,4\}$. Then there is
  a cut vertex $v$ of $G'$ such that $a_1, \dots, a_4$ are in four distinct
  components of $G' - v$.
\end{lemma}

\begin{proof}
 Let us consider an auxiliary graph $G''$ which is obtained from $G'$ by adding two new vertices $x$, $y$
 and attaching $x$ to $a_1,a_2$ and $y$ to $a_3,a_4$.
 By the assumptions, $G''$ is connected and moreover, there are no two vertex-disjoint paths connecting $x$ and $y$.
 By Menger's theorem (see, e.g.,~\cite[Corollary~3.3.5]{diestel10}), there is a
 cut-vertex $v \in V(G'')\setminus\{x,y\}=V(G')$ disconnecting $x$ and $y$. 
 Let $C_1$ be the connected component of $G'' - v$ containing $x$ and $C_2$ be the component containing $y$.
 Let $C'_i$, for $i=1,2$, be the subgraph of $G'$ induced by $v$ and the vertices of $C_i \cap G'$.
 Note that, since $G'$ is connected, both $C'_1$ and $C'_2$ are connected.
 We show that $v$ is the desired cut vertex. 
  
 Let  $p_1$ be an $a_1 a_2$-path in  $C'_1$
 and $p_2$ an $a_3 a_4$-path in $C'_2$. Since $C'_1$ and $C'_2$ are connected, 
 such paths $p_1$ and $p_2$ exist.
 Moreover, $p_1$ and $p_2$ may intersect only in $v$;
 however, according to the assumptions, they have to intersect in a vertex outside $A$.
 Therefore, they must
 intersect in $v$ and $v\notin A$.
 Overall, we have verified that any $a_i a_j$-path passes through $v$, for $1
 \leq i < j \leq 4$, which shows that $v$ is the desired cut vertex.
\ifconf
\qed
\fi
\end{proof}

\begin{proof}[Proof of Lemma~\ref{l:cut_vertex}.]
  We apply Lemma~\ref{l:intersecting_paths} to $B$ and to  $A = \{ a, b, c, d \}$.
  Let us consider a valid labelling $L$ of $B$. Up to swapping the labels, we
  may assume that $L(a) = L(c) = \{1\}$ and $L(b) = L(d) = \{0\}$. Then
  Proposition~\ref{prop:arrows_labels} together with
  Lemma~\ref{l:single_labels} imply that any proper $ab$, $bc$, $cd$, or
  $ad$-walk is nontrivial, whereas any proper $ac$ or $bd$-walk is trivial.
  Then, the assumptions of Lemma~\ref{l:intersecting_paths} are satisfied due
  to Lemmas~\ref{l:noninterleaving_walks_intersect}
  and~\ref{l:interleaving_walks_intersect}.
\ifconf
\qed
\fi
\end{proof}

\paragraph{Inside split triangle.}
Finally, we prove Proposition~\ref{p:redrawings} for
an inside split triangle. 
\begin{proof}[Proof of Proposition~\ref{p:redrawings} for an inside split triangle]
Let $a$, $b$, $c$ be the three vertices of $Z$ from the
definition of the inside split triangle; see Definition~\ref{d:redrawable} or Fig.~\ref{f:redrawable}.

First, similarly as in the proof for inside squares, we take care of trivial
inside bridges via suitable labellings and Lemma~\ref{l:labels}. We reach a
projective HT-drawing $(D_1, \lambda_1)$ still satisfying the
assumptions of Proposition~\ref{p:redrawings}, which in addition satisfies
$\lambda_1(e) = 0$ for any edge $e$ of $\gin$ that does not belong to a
nontrivial bridge.

Now, let us consider nontrivial inside bridges. By the assumptions, each such
bridge is either an
\emph{$a$-bridge}, that is, a nontrivial inside bridge which contains
$a$ (and $b$ or $c$ or both), or a \emph{$bc$-bridge} which contains $b$ and $c$,
but not $a$. We consider valid labellings of these bridges. As usual, we
may swap all labels in a valid labelling when needed.
This way, it is easy to
check that every $a$-bridge $B$ admits a valid labelling $L_B$ such that
$L_B(a) = \{1\}$, whereas all other labels are $0$. Similarly, each
$bc$-bridge $B$ admits a valid labelling $L_B$ such that
$L_B(b) = \{1\}$ and $L_B(c) = \{0\}$. We apply Lemma~\ref{l:labels} and 
we reach a projective HT-drawing $(D_2, \lambda_2)$ still satisfying the
assumptions of Proposition~\ref{p:redrawings}, which in addition satisfies the
following property. The edges $e$ of $\gin$ with $\lambda_2(e) = 1$ are exactly
the edges of an $a$-bridge which are incident to $a$ or edges of a
$bc$-bridge incident to $b$.

If we do not have any $bc$-bridge, then all nontrivial edges are
incident to $a$ and we finish the proof by setting $D' = D_2$ and
letting $\lambda'$ be identically $0$, similarly as in the cases of an inside fan or
an inside square. However, if we have $bc$-bridge(s), we need to be more
careful. 

Let $E_a^x$ and $E^x_{bc}$ be the sets of edges incident to a vertex $x$ in
an $a$-bridge and the set of edges incident to $x$ in 
a $bc$-bridge, respectively. Because $D_2$ is a projective HT-drawing,
we have $\lambda_2(e)\lambda_2(f) = \crno_{D_2}(e,f)$ for any pair of
independent edges $e$ and $f$. In particular, $\crno_{D_2}(e,f) = 1$ for a pair of
independent edges if and only if one of the edges belongs to $E_a^a$ and the
second one to $E_{bc}^b$.

Now, for every edge $e \in E^b_{bc}$, we perform the vertex-edge switch over each
vertex different from $a$, $b$, $c$ of each $a$-bridge obtaining a drawing
$D_3$. We perform the
switches inside $S^+$. This way, we change the crossing number of such $e$ with
edges from $E^a_a$, $E^b_a$ and $E^c_a$. In particular, after this redrawing,
we get $\crno_{D_3}(e,f) = 1$ for a pair of
independent edges if and only if one of the edges belongs to $E_a^c$ and the
second one to $E_{bc}^b$. See Fig.~\ref{f:a_bc_bridges}.

Finally, for every edge $e \in E_a^c$, we perform the vertex-edge switch over
each vertex different from $b$ and $c$ of each $bc$-bridge obtaining the 
final drawing $D'$. Again, we perform the switches inside $S^+$.  This way,
we change the crossing number of such $e$ with
edges from $E_{bc}^b$ and $E^c_{bc}$. However, it means that  $\crno_{D'}(e,f)
= 0$ for any pair of independent edges. That is, $D'$ is the required
ordinary HT-drawing on $S^2$. See Fig.~\ref{f:a_bc_bridges}.
\ifconf
\qed
\fi
\end{proof}

\begin{figure}
\begin{center}
  \includegraphics{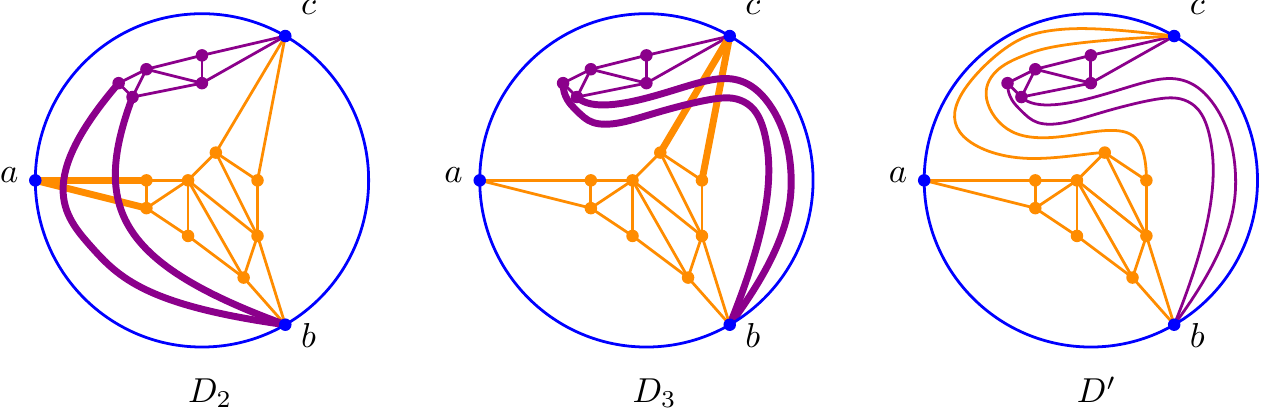}
  \caption{An example of redrawing an inside split triangle with one $a$-bridge and one
  $bc$-bridge. The edges participating in independent pairs crossing oddly
are thick. For simplicity of the picture, the drawings $D_3$ and $D'$ are
actually simplified. For example, the vertex-edge switches used to obtain
$D_3$ from $D_2$ introduce many pairs of independent edges crossing evenly and some
pairs of adjacent edges crossing oddly. These intersections are removed in the
picture as they do not play any role in the argument. (In particular, the
drawing $D'$ is, in fact, typically not a plane drawing.)}
  \label{f:a_bc_bridges}
\end{center}
\end{figure}

\section{Redrawing by Pelsmajer, Schaefer and \v{S}tefankovi\v{c}}
\label{s:PSS_in_out}

It remains to prove Theorem~\ref{t:PSS_in_out}. As mentioned above, our proof
is almost identical to the proof of Theorem~2.1~in~\cite{PSS07}. The only
notable difference is that we avoid contractions.\footnote{Our reason why we
avoid contractions is mainly for readability issues. Contractions yield
multigraphs and, formally speaking, we would have to redo several notions for
multigraphs. Introducing multigraphs in the previous sections would be disturbing and
it is not convenient to repeat all the definitions in such setting now.}
As noted before, the proof of Lemma~3 in~\cite{Fulek2012AdjCrossings} can also be extended to yield the desired result.

\begin{proof}
First, we want to get a drawing such that there is only one edge of $Z$
which may be intersected by other edges. Here, part of the argument is almost
the same as the analogous argument in the proof of Lemma~\ref{lem:simple}.

Let us consider an edge $e = uv \in E(Z)$
intersected by some other edges and let $f = vw \in E(Z)$ be a neighbouring edge of
$e$. 
We again almost-contract $e$ so that we move the vertex $v$ towards $u$ until we
remove all intersection of $e$ with other edges. This way, $e$ is now free of
crossings and these crossings appear on $f$. Since both $e$ and $f$ were even
edges in the initial drawing, $f$ remains even after the redrawing as well.
Finally, since we want to keep the position of $Z$, we consider a
self-homeomorphism of $S^2$ which sends $v$ back to its original position. See
Fig.~\ref{f:almost_contract}.

By such redrawings, it can be achieved that only one edge $e_0 = u_0v_0$ of $Z$ may be
intersected by other edges while keeping $Z$ fixed and $e_0$ even. Without
loss of generality, we may assume that the original drawing $D$ satisfies these
assumptions.

\begin{figure}
\begin{center}
\includegraphics{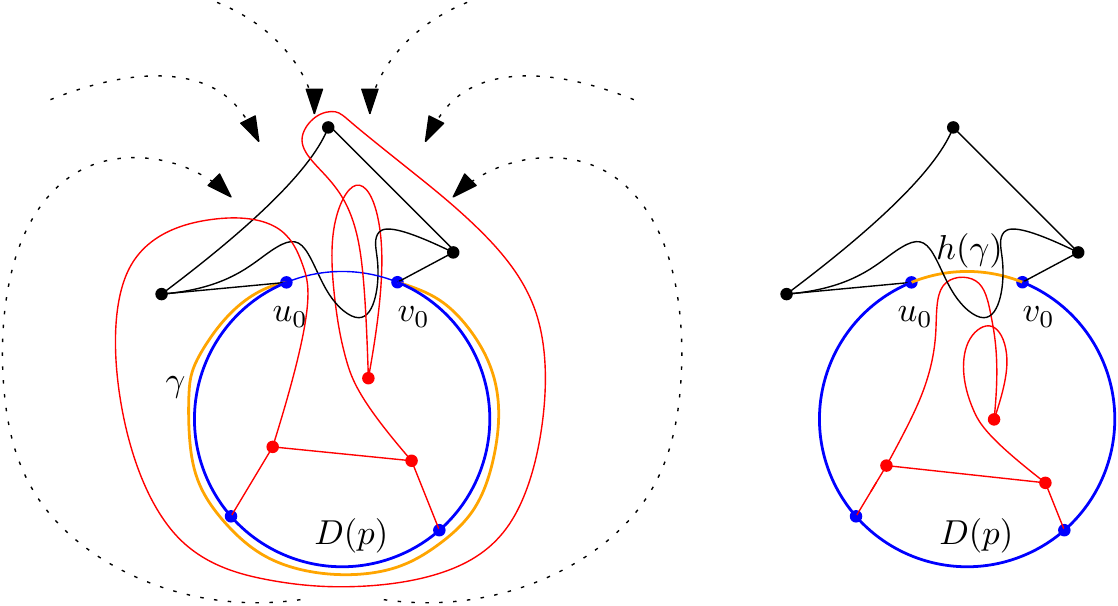}
  \caption{An illustration of the self-homeomorphism $h$, which maps $B$ to $S^+$, applied to the drawing of $\gin - e_o$.}
\label{f:PSS_deformation}
\end{center}
\end{figure}

Let $p$ be the path in $Z$ connecting $u_0$ and $v_0$ avoiding $e_0$.
Let us also consider an arc $\gamma$ connecting $u_0$ and $v_0$ outside (that is in
$S^-$) close to $D(p)$ such that it does not cross any inside edge. The closed arc
obtained from $\gamma$ and $D(p)$ bounds two disks ($2$-balls). Let $B$ be the
open disk which contains $S^+$. Finally, we consider a self-homeomorphism
$h$ of $S^2$ that keeps $D(p)$ fixed and maps $B$ to $S^+$. Considering the
drawing $h \circ D$ on $\gin - e_0$, it turns out that $\gin - e_0$ is now drawn in $S^+$,
up to $p$, which stays fixed. For the edge $e_0$, we also keep its original
position, that is, we do not apply $h$ to this edge. See Fig.~\ref{f:PSS_deformation}.

Since the redrawing is done by a
self-homeomorphism, we do not change the number of crossings among pairs of
edges in $\gin$.  Analogously, we map $\gout$ to $S^-$ and we get the
required drawing.
\ifconf
\qed
\fi
\end{proof}
\fi

\ifconf
\else

\section*{Acknowledgment} We would like to thank Alfredo Hubard for fruitful
discussions and valuable comments.

\bibliographystyle{alpha}
\bibliography{citations}
\addcontentsline{toc}{section}{References}
\fi

\end{document}